\documentclass[twocolumn,twoside,pra,aps,showpacs]{revtex4}
\usepackage{amsmath,mathrsfs,amsbsy,color,graphicx,bm,amsthm,amsfonts,dsfont}
\usepackage{times,wasysym,amssymb}
\usepackage{hyperref}
\usepackage{units}
\usepackage[british]{babel}

\newenvironment{psmallmatrix}
  {\left(\begin{smallmatrix}}
  {\end{smallmatrix}\right)}

\begin{document}
\theoremstyle{remark}
\newtheorem{theorem}{Theorem}
\newtheorem{problem}{Problem}
\newtheorem{proposition}{Proposition}
\newtheorem{example}{Example}

\title{Qudit quantum computation on matrix product states with global symmetry}

\author{Dong-Sheng Wang}
\affiliation{Department of Physics and Astronomy, University of British Columbia,  Vancouver, BC, Canada}
\author{David T. Stephen}
\affiliation{Department of Physics and Astronomy, University of British Columbia,  Vancouver, BC, Canada}
\author{Robert Raussendorf}
\affiliation{Department of Physics and Astronomy, University of British Columbia,  Vancouver, BC, Canada}

\begin{abstract}
Resource states that contain nontrivial symmetry-protected topological order
are identified for universal single-qudit
measurement-based quantum computation.
Our resource states fall into two classes:
one as the qudit generalizations of the 1D qubit cluster state,
and the other as the higher-symmetry generalizations of the spin-1
Affleck-Kennedy-Lieb-Tasaki (AKLT) state,
namely, with unitary, orthogonal, or symplectic symmetry.
The symmetry in cluster states protects information propagation (identity gate),
while the higher symmetry in AKLT-type states enables nontrivial gate computation.
This work demonstrates a close connection
between measurement-based quantum computation and symmetry-protected topological order.
\end{abstract}


\date{\today}
\maketitle


\section{Introduction}
\label{sec:intr}

Measurement-based quantum computation~\cite{RB01} (MBQC)
is an alternative of the circuit model~\cite{NC00},
and it has the advantage that
once a certain resource state is available,
universal quantum computation can be executed
by local non-entangling projective measurements.
In recent years, new types of universal resource states have been identified~\cite{GE07,CZG+09,CDJ+10,WAR11,WAR12,Miy11,DBB12,WHR14,NW15,MM15b}
beyond the usual cluster states and graph states~\cite{HDE+06}.
In particular,
states which are unique
ground states of local Hamiltonians can be
easily prepared by cooling.
Among these novel states,
the AKLT-type states, which have non-trivial
bosonic symmetry protected topological (SPT) orders~\cite{CGW11,SPC11,CGL+13,DQ13a,DQ13b},
have triggered the investigation of highly symmetric ground states~\cite{NMC+13,GHG+10,NLC+16}.

In this work we explore the connection between MBQC
and symmetric quantum states with nontrivial SPT orders.
We mainly study the computational properties of two classes of states:
one as the qudit generalizations of the
one-dimensional (1D) cluster state,
and the other of the AKLT state.
For this purpose,
we employ the correlation space picture~\cite{VC04,GE07} of MBQC
where resource states are represented as matrix product states (MPS)~\cite{PVW+07}.
For cluster states, we define 1D qudit cluster states
that have different bosonic
SPT orders according to the symmetry $\mathbb{Z}_d\times \mathbb{Z}_d$,
as the generalizations of 1D qubit cluster state
and the qudit cluster state defined before~\cite{ZZX+03,Cla06}.
We present a detailed analysis of their properties,
including their symmetry, parent Hamiltonian,
representation by quantum circuits,
and nontrivial gate computation.
For AKLT-type states, we show that there are universal resource states
for the $SU(N)$~\cite{GR07,KHK08,MUM+14},
$SO(N)$~\cite{TZX08}, and $Sp(2n)$~\cite{RS91,WZ05,FC09}
symmetric generalizations of AKLT states,
and that they have the appealing features including
nontrivial gates induced by their symmetry
and the ability to prepare and readout the virtual space via projection.
Our study also demonstrates that,
for both 1D cluster states and AKLT-type states,
the correlation space computation is equivalent to the real space computation of MBQC.

Our work provides families of 1D resource states that have nontrivial SPT orders.
On one hand, this extends some previous study, e.g., Ref.~\cite{GE10},
which focused on the cases with both on-site and virtual system of dimension two,
and not in the context of SPT orders.
On the other hand, the resource states we have identified can be viewed as representative points
in the domain of the corresponding SPT phases,
hence serve as a stepping-stone for the understanding of computational properties of
SPT phases with high symmetries,
which so far only have been explored for the cases of finite groups
and $SU(2)$~\cite{BM08,DB09,ESB+12,PW15,MM15}.

This paper contains the following parts.
In Section~\ref{sec:review} we briefly review
MBQC on the qubit cluster state and spin-1 AKLT state,
highlighting the role of symmetry
and the equivalence between the virtual space picture and real space picture for both cases.
We study the qudit cluster states in Section~\ref{sec:clu-d},
and focus on identifying their SPT orders and computational properties.
The AKLT-type states with unitary, orthogonal, and symplectic symmetries
are studied in Section~\ref{sec:gaklt},
and we find there exist large classes of universal resource states.
In Section~\ref{sec:conc} we conclude.

\section{Cluster state and AKLT state}
\label{sec:review}

In this section, we review MBQC based on the 1D qubit cluster state and spin-1 AKLT state
both in the real space picture and virtual space picture,
which is necessary to understand their generalizations in later sections.

\subsection{MPS circuit and correlation space}
\label{subsec:mps}

We start from MPS theory~\cite{PVW+07,Sch11,Oru14}
following the methodology of correlation space quantum computation~\cite{GE07}.
An arbitrary finite-dimensional quantum state can be written as
\begin{equation}\label{}
  |\Psi\rangle=\sum_{i_1,\dots,i_N=0}^{d-1}\langle R|
A^{[1]}_{i_1} \cdots A^{[N]}_{i_N} |L\rangle |i_1\dots i_N\rangle
\end{equation}
for open boundary condition (OBC) case with boundary states $\langle R|$ and $|L\rangle$
for $N$ local systems (spins) with local dimension $d$.
The tensor network representation of a MPS is shown in Fig.~\ref{fig:MPScir}(a).
The states $\langle R|$ and $|L\rangle$ live in,
and the operators $A^{[n]}_{i_n}$ ($n=1,\dots,N$) act on an ancillary space of dimension $\chi$.
This space is also known as the bond space, correlation space, or the virtual space,
and $\chi$ is often known as the bond dimension.

\begin{figure}[t!]
  \centering
  \includegraphics[width=.45\textwidth]{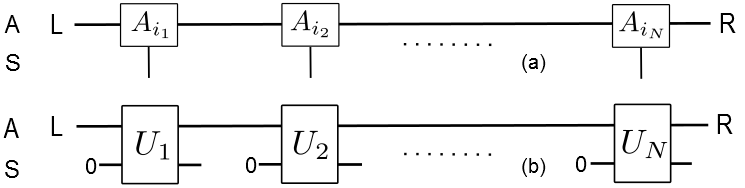}\\
  \caption{(a) The tensor network representation of a matrix product state.
  The horizontal line is for ancilla (correlator) A,
  and the vertical legs are for the system S.
  (b) The quantum circuit representation of a matrix product state.
  The system S containing $N$ spins is initially at state $|00\cdots 0\rangle$,
  and the ancilla (correlator) A is initially at state $|L\rangle$.
  Each unitary operator $U_n$ acts on the ancilla and one system spin.
  }\label{fig:MPScir}
\end{figure}

In this work we consider translation-invariant systems,
so the site superscripts can be dropped.
The set of Kraus operators $\{A_{i_n}\}$ on each site forms a quantum channel $\mathcal{E}_n$ with
the trace preserving condition $\sum_{i_n} A_{i_n}^\dagger A_{i_n}=\mathds{1}$
following from the normal form of MPS~\cite{Sch11}.
From the dilation theorem~\cite{Sti55}, a channel $\mathcal{E}_n$ can be realized by
a unitary operator $U_n$ acting on a larger space.
As a result, a MPS can be represented by a quantum circuit~\cite{SSV+05},
shown in Fig.~\ref{fig:MPScir}(b),
wherein each unitary operator $U_n$ is defined such that
\begin{equation}\label{}
  U_n |0\rangle=\sum_{i_n} |i_n\rangle \otimes A_{i_n}.
\end{equation}
\noindent The projection $|R\rangle\langle R|$ on the ancilla results in the MPS $|\Psi\rangle$,
while tracing out the system results in a sequence of quantum channels $\mathcal{E}_n$
acting on the correlator such that
\begin{equation}\label{}
\mathcal{E}_n(\rho)=\sum_{i_n} A_{i_n} \rho A_{i_n}^\dagger.
\end{equation}
Note that for translation-invariant systems $\mathcal{E}_n$ and $U_n$ are both site-independent.

The above framework provides a natural starting point to understand MBQC.
For MBQC on cluster states, 1D wires can be cut out to represent qubits,
while junctions among wires enable entangling gates.
A 1D cluster state can be expressed as a MPS,
so the projective measurement on each spin in it
can be interpreted in the correlation space~\cite{GE07}.
In general, given a 1D resource state, a
projective measurement on site $n$
in the standard basis $\{|i_n\rangle\}$
induces the set of operators $\{A_{i_n}\}$ on the correlation space,
while measurement in a basis rotated by a unitary operator $W_n$ leads to operators
\begin{equation}\label{}
  \tilde{A}_{i_n}=\sum_{j_n} w_{i_nj_n} A_{j_n}
\end{equation}
for $\langle i_n|W_n=\langle w_{i_n}|=\sum_{j_n} w_{i_nj_n} \langle j_n|$.
This leads to the so-called real space picture and virtual space picture descriptions of MBQC.
Namely, in the real space picture, the computation is carried out on the physical systems by measurements,
while in the virtual space picture, the computation is carried out on the virtual space.
The projective measurements are selective,
i.e., each outcome from a projector is recorded instead of mixing together with each other,
and the operators $\tilde{A}_{i_n}$ enacted should be unitary up to byproduct,
which further should be able to be pulled out and correctable after the computation.
A sequence of projective measurements can then be interpreted in the correlation space
as the action of a sequence of unitary gates.

Furthermore, for translation-invariant MPS with an on-site symmetry $G$, e.g.,
the cluster states and AKLT-type states studied in this work,
the following symmetry condition holds
\begin{equation}\label{eq:symcond}
  \sum_j u_{ij}(g) A_j=V(g)^\dagger A_i V(g)
\end{equation}
for $U(g)=(u_{ij})(g)$ as a linear unitary representation of $g\in G$,
and $V(g)$ as a projective unitary
representation of $g$~\cite{FNW92,SWP+09,CGW11,SPC11}.
Note that to keep a MPS invariant,
an on-site symmetry operation involves a symmetry operator on each site in the system.
The symmetry condition~(\ref{eq:symcond}) basically means
a projective measurement in a rotated basis
by a symmetry operation $U(g)$ can be equivalently understood as
a conjugation by $V(g)$ on the correlation space.
In particular, the relation~(\ref{eq:symcond}) is employed for
the byproduct propagation for cluster states with on-site $\mathbb{Z}_d\times \mathbb{Z}_d$ symmetry for a certain $d$,
and gate computation for AKLT-type states with on-site Lie group symmetries.

\subsection{Qubit cluster state}
\label{subsec:clu}

We now apply the two pictures described above to the 1D qubit cluster state.
We will emphasize the equivalence between the two pictures,
which has not been explicitly demonstrated before.
Such an equivalence also holds for AKLT states, as will be seen later on.

To prepare a qubit cluster state, each qubit is initially in the state $|+\rangle=\frac{1}{\sqrt{2}}(|0\rangle+|1\rangle)$,
and the two-qubit controlled-phase gate $CZ=[\mathds{1},0;0,Z]$ is applied on each neighboring pair of qubits.
For information propagation, i.e., identity gate along the wire,
the following relation holds
\begin{equation}\label{}
  \langle s|CZ|\text{in}\rangle|+\rangle=HP_s|\text{in}\rangle
\end{equation}
for an unknown input state $|\text{in}\rangle$ on the first qubit,
projector $P_s=|s\rangle\langle s|$ ($s=0,1$) in the $Z$ basis on the first qubit,
and Hadamard gate $H$.
Projection in the $X$ basis then leads to
\begin{equation}\label{}
  \langle s|(H\otimes\mathds{1})CZ|\text{in}\rangle|+\rangle=HZ^s|\text{in}\rangle.
\end{equation}
A further projection in the $X$ basis on the second qubit leads to
\begin{align}\label{eq:teleport}\nonumber
  &\langle s|\langle t|(H\otimes H\otimes\mathds{1})(CZ\otimes\mathds{1})(\mathds{1}\otimes CZ)|\text{in}\rangle|+\rangle|+\rangle \\
  &=HZ^t HZ^s |\text{in}\rangle= X^t Z^s |\text{in}\rangle.
\end{align}
This shows that an unknown state $|\text{in}\rangle$ on the first qubit
is teleported to the third qubit with Pauli byproduct $X^t Z^s$
(also see Refs.~\cite{GC99,ZLC00}).

When expressing the cluster state as a MPS for the virtual space interpretation,
the quantum circuit to generate it (according to Fig.~\ref{fig:MPScir}) can be easily obtained as follows:
each unitary $U$ in the circuit is site-independent and takes the form
\begin{equation}\label{}
  U=CZ \cdot \textsc{swap}
\end{equation}
for the qubit swap gate $\textsc{swap}|ij\rangle=|ji\rangle$.
The ancilla A is in the unknown initial state $|\text{in}\rangle$,
while each qubit in S is initialized at state $|+\rangle$ instead of $|0\rangle$.
The operators $X^t Z^s$ become the Kraus operators in the virtual space picture
after blocking each two sites together.

To execute a general qubit gate,
a sequence of measurements in rotated bases~\cite{RB01}
leads to
\begin{equation}\label{}
  |\text{out}\rangle=HZ(\alpha_4)Z^{s_4} HZ(\alpha_3)Z^{s_3} HZ(\alpha_2)Z^{s_2}HZ^{s_1}|\text{in}\rangle
\end{equation}
for both of the pictures.
Here, the rotated basis on the $i$-th qubit is by gate $Z(\alpha_i)$ for
$Z(\alpha):=e^{-i\alpha Z}=[e^{-i\alpha},0;0,e^{i\alpha}]$.
This shows the equivalence of the real and virtual pictures for the qubit cluster state
for both information propagation and gate computation.

\subsection{Spin-1 AKLT state}
\label{subsec:1aklt}

For a spin-1 chain the AKLT ground state with $SO(3)$ symmetry~\cite{AKLT87,AKLT88}
has been shown to be a resource state.
In translation-invariant MPS form it can be described by the
three Pauli operators (with coefficient $\sqrt{1/3}$ )
$X$, $Y$, and $Z$ with bond dimension $\chi=2$~\cite{FNW92,SWP+09}.
The symmetry $SO(3)$ is represented fundamentally on the spin-1 for each physical site,
and represented projectively on the virtual spin-1/2 in the well-known valence bond state picture.

For our purpose, we first find a unitary MPS circuit to represent
the AKLT state as follows.
Each local spin-1 in the system is encoded by two qubits and initialized in the state
\begin{equation}\label{}
|+_3\rangle:=\frac{1}{\sqrt{3}} (|01\rangle+|10\rangle+|11\rangle).
\end{equation}
The unitary operator $U$ in the circuit is also site-independent
and is a uniformly-controlled gate with the form
\begin{align}\label{eq:Uaklt}
  U=P_{00}\otimes \mathds{1}+ P_{10}\otimes X+ P_{01}\otimes Z+ P_{11}\otimes Y,
\end{align}
with the ancilla as the target and the local spin as the control,
and $P_{s_1s_2}=|s_1s_2\rangle\langle s_1s_2|$
denote projectors on a local spin.
The ancilla carries the initial unknown state $|\text{in}\rangle$,
while projective measurement on a local spin with outcomes
$s_1s_2\in\{01,10,11\}$ yields gates $X^{s_1}Z^{s_2}$
acting on the input state $|\text{in}\rangle$ as the byproduct,
and the state $|\text{in}\rangle$ is then teleported to the next site.
Measurement in a rotated basis, e.g.,
by an orthogonal operator $R_x$
(corresponding to a qubit rotation $X(\theta)$ due to $SU(2)/\mathbb{Z}_2\cong SO(3)$),
induces a nontrivial gate according to the symmetry condition
\begin{equation}\label{}
  \sum_j r_{ij}(g) A_j=V(g)^\dagger A_i V(g)
\end{equation}
for $R(g)=(r_{ij})(g)$ as the linear unitary representation of $g\in SO(3)$,
and $V(g)\in SU(2)$ as the projective unitary representation of $g$.
An arbitrary qubit gate can be induced by measurements in rotated bases on three sites
and yields
\begin{align}\label{eq:out-aklt}\nonumber
  |\text{out}\rangle=&X(-\gamma)X^{r_1}Z^{r_2}X(\gamma) Z(-\beta)X^{t_1}Z^{t_2}Z(\beta) \\ &X(-\alpha)X^{s_1}Z^{s_2}X(\alpha)|\text{in}\rangle.
\end{align}
As is previously known~\cite{BM08},
the success probability of a gate is $2/3$
since, e.g.,
$X(-\alpha)X^{s_1}Z^{s_2}X(\alpha)$ equals $X(2\alpha)$ for
$s_1s_2=11$ or $01$,
and $\mathds{1}$ for $s_1s_2=10$.
When the recorded measurement outcomes $s_1s_2$ shows that
an identity gate is induced, the measurement is repeated on the next site until a
success.
This general feature will be generalized for AKLT states with higher symmetries
in Section~\ref{sec:gaklt}.

Now, similar with the cluster state, we can obtain
the circuit to represent AKLT state in real space picture by employing swap gates.
Each spin-1 is represented by the three-dimensional subspace of two qubits and has the initial state $|+_3\rangle$.
For each site $n$, let us label the two qubits as $n_1$ and $n_2$.
Then a swap gate is applied between $n_2$ and $(n+1)_2$,
followed by the uniformly-controlled gate $U$~(Eq.~\ref{eq:Uaklt}) acting on $n_2$, $(n+1)_1$, and $(n+1)_2$,
with $(n+1)_2$ as the target.
Note that these entangling gates between different sites do not commute,
which means AKLT state is on a \emph{directed} graph,
different from the cluster state case.
With this circuit construction, it is immediate to see the equivalence
between the real and virtual pictures for MBQC on AKLT state.

We have seen above that for both the cluster state and AKLT state
the byproduct are Pauli operators.
This is actually a manifestation of their symmetries:
$\mathbb{Z}_2\times \mathbb{Z}_2$ for cluster state and $SO(3)$ for AKLT state.
In the following we generalize the connection between
1D bosonic SPT order and MBQC
based on the recent
SPT phase classifications~\cite{CGW11,SPC11,CGL+13,DQ13a,DQ13b}.

\section{Qudit cluster states MBQC}
\label{sec:clu-d}

In this section we construct qudit cluster states that have different SPT orders.
Our study shows that the computational power of a qudit cluster state
is indicated by its bond dimension $b$, namely, the gates induced by projective measurements
form the whole group $SU(b)$.
This also means cluster states with different SPT orders but
with the same bond dimension have the same computational power.

\subsection{Construction of states}

For qudits, the Pauli operators are generalized to the
Heisenberg-Weyl operators $X^jZ^k$ with
\begin{subequations}
\begin{alignat}{2}\label{}
  X^j= & \sum_\ell |\ell\rangle\langle\ell+j|, \\
  Z^k= & \sum_\ell \omega^{k\ell} |\ell\rangle\langle\ell|,
\end{alignat}
\end{subequations}
which satisfy
\begin{equation}\label{eq:hwc}
  X^j Z^k=\omega^{jk} Z^k X^j,
\end{equation}
for $\omega=e^{i2\pi/d}$, $j,k\in \mathbb{Z}_d$.

We construct qudit cluster states using the generalized controlled-$Z$ gate
\begin{equation}\label{eq:Sx}
  S^x:=\sum_\ell P_\ell \otimes Z^{x\ell}
\end{equation}
for $x=1,2\dots,d-1$.
To prepare a 1D qudit cluster state,
we let each qudit spin be in initial state
$|+_d\rangle:=\frac{1}{\sqrt{d}}\sum_\ell |\ell\rangle$ as the analog of the qubit case. Now we can define different cluster states
by the alternating application of gates $S^x$ and $S^y$ along the wire for $x+y=d$.
Each qudit cluster state is denoted by $|C_d(x,y)\rangle$ in this work,
and it is shown later that they have different SPT orders.

The stabilizer operators and parent Hamiltonian
for each $|C_d(x,y)\rangle$ can be easily obtained.
The way to derive the qudit stabilizers follows the same procedure as for the qubit case.
If we denote the $CZ$ gate by $S$ and the qubit cluster state by $|C_2\rangle$,
the stabilizers are derived by the following relations
 \begin{equation}\label{}
   SXS^\dagger|C_2\rangle=|C_2\rangle,\; SX[a]S^\dagger=X[a]\otimes_b Z[b]
 \end{equation}
for $a,b$ as the label of sites and $b$ as the nearest neighbor of $a$.
For the qudit case, the lattice is bipartite, containing even sites and odd sites.
An even site is first acted on by $S^x$ then by $S^y$, an the opposite for an odd site.
We find
\begin{align}\label{}
  S^x X[a] S^y  =& X[a] \otimes_b Z[b]^{-x},
\end{align}
and also its hermitian conjugate.
For each state $|C_d(x,y)\rangle$ the generating stabilizers take the form
$Z^{x}XZ^{y}$ for odd sites, and $Z^{y} X Z^{x}$ for even sites.

The so-called parent Hamiltonian, denoted by $H_x$ for $|C_d(x,y)\rangle$,
follows directly from the sum of stabilizers and their hermitian conjugate,
with some minor differences for PBC and OBC cases (e.g., Ref.~\cite{ZZX+03}).
For OBC, the parent Hamiltonian cannot include the boundary terms since they break the symmetry,
which means there will be a free edge state at the boundary.
For PBC, the state $|C_d(x,y)\rangle$ is the same as $|C_d(y,x)\rangle$ (up to one-site translation),
e.g., $|C_5(1,4)\rangle$ is the same as $|C_5(4,1)\rangle$, $|C_5(2,3)\rangle$ is the same as $|C_5(3,2)\rangle$,
while $|C_5(1,4)\rangle$ is \emph{not} the same as $|C_5(2,3)\rangle$.
Therefore, for PBC a qudit cluster state is the unique ground state of its parent Hamiltonian.

\subsection{\texorpdfstring{$\mathbb{Z}_d\times \mathbb{Z}_d$}{Lg} symmetry \& SPT order}

Here we identify the symmetry and SPT order of the qudit cluster states we defined.
As demonstrated in Refs.~\cite{CGW11,SPC11,CGL+13},
1D bosonic SPT phases protected by an on-site symmetry group $G$
are labelled by elements of the second cohomology group.
According to
\begin{equation}\label{}
H^2(\mathbb{Z}_d\times \mathbb{Z}_d,U(1))=\mathbb{Z}_d,
\end{equation}
there are $d$ phases with one trivial
protected by the symmetry $\mathbb{Z}_d\times \mathbb{Z}_d$.
Suppose the generators of $\mathbb{Z}_d\times \mathbb{Z}_d$
are $(g,e)$ and $(e,h)$ for $g^d=e$, $h^d=e$.
Now different projective representations $v_x$ can be introduced
\begin{equation}\label{eq:rhox}
  v_x: (g_i,h_j) \mapsto X^i Z^{xj},
\end{equation}
for $x=1,2,\dots, d-1$.
From
\begin{equation}\label{}
  v_x((g_i,h_j))v_x((g_a,h_b))=\omega_x(ij,ab)v_x((g_{i+a},h_{j+b}))
\end{equation}
and Eq.~(\ref{eq:hwc}) we can find $d-1$ different 2-cocycles
\begin{equation}\label{}
  \omega_x(ij,ab)=\omega^{-xaj},
\end{equation}
and each SPT phase can be labeled by a single integer $x$.
Among those nontrivial phases,
phases with $\text{gcd}(x,d)=1$, i.e., $x$ and $d$ are coprime,
are maximally noncommutative (MNC)~\cite{BZ98,ESB+12,EBD13} with cocycle $\omega_x$,
which essentially means
there exists a unique projective irreducible representation (irrep) of the symmetry
$\mathbb{Z}_d\times \mathbb{Z}_d$ with dimension $d$.
Also for phases that are not MNC,
the dimension of a projective irrep of the symmetry is $d/\text{gcd}(x,d)$.
It turns out the MNC condition is crucial for our study of qudit cluster states.

To show the SPT order of qudit cluster states,
we first need to introduce generalized Fourier operators
\begin{equation}\label{eq:fk}
  F_k:=\frac{1}{\sqrt{d}}\sum_{j\ell } \omega^{kj\ell } |j\rangle\langle \ell|
\end{equation}
for $\text{gcd}(k,d)=1$ and $1\leq k\leq d-1$,
which are all unitary
\begin{equation}\label{}
  F_k^\dagger F_k = F_k F_k^\dagger=\mathds{1},
\end{equation}
and also satisfy
\begin{equation}\label{eq:permutation}
  F_k^4=\mathds{1},\; F_k^2=\sum_\ell |\ell\rangle\langle-\ell|\equiv \Pi,
\end{equation}
which is a permutation independent of $k$.
The following properties also hold
\begin{equation}\label{}
  F_k|j\rangle=Z^{kj}|+_d\rangle,\; F_k^\dagger|j\rangle=Z^{-kj}|+_d\rangle,
\end{equation}
and
\begin{equation}\label{}
F_kZ^jF_k^\dagger=X^{-(Nd-j)/k}, F_k^\dagger Z^{kj}F_k=X^{-j},
\end{equation}
for an integer $N$ s.t. $(Nd-j)/k\in \mathbb{Z}_d$ mod $d$.
The standard Fourier operator $F$ is recovered as a special case for $k=1$.

Now consider a building block of three qudits with the first in some unknown state $|\psi\rangle$,
after the two entangling gates the resulting state is
$|\Psi\rangle:=(\mathds{1}\otimes S^y)(S^x\otimes\mathds{1})|\psi\rangle|+_d\rangle|+_d\rangle$.
The following relation holds
\begin{align}\label{eq:qudpro1}
  \langle s|\langle t|(X^i\otimes X^j)|\Psi\rangle
  =F_y P_{t-j} F_x P_{s-i} |\psi\rangle.
\end{align}
Also note $F_x=F_y^\dagger$.
Two successive projections in a basis rotated by Fourier operators lead to
\begin{align}\label{}
  \langle s|\langle t|(F_\alpha\otimes F_\beta\otimes \mathds{1})|\Psi\rangle
  =F_y Z^{\beta t} F_x Z^{\alpha s} |\psi\rangle,
\end{align}
which is a generalization of teleportation in Eq.~(\ref{eq:teleport})
to 1D qudit cluster states we defined.

Now we can demonstrate the symmetry and SPT order of 1D qudit cluster states
by the following proposition.
\begin{proposition}\label{prop:clustspt}
    A qudit cluster state $|C_d(x,y)\rangle$
    and its parent Hamiltonian $H_x$ have the
    on-site $\mathbb{Z}_d\times \mathbb{Z}_d$ symmetry
    after blocking two sites together
    with $x$ labeling its SPT order protected by $\mathbb{Z}_d\times \mathbb{Z}_d$.
\end{proposition}
\begin{proof}
    The symmetry follows from products of stabilizers for all even or odd sites.
    Denoting the two generators of $\mathbb{Z}_d\times \mathbb{Z}_d$ as $g$ and $h$,
    the on-site linear unitary representations are
    $u(g)=X \otimes \mathds{1}$ and $u(h)=\mathds{1}\otimes X$.

    To show the SPT order, we check the projective representation of the symmetry.
    First, for the case $\text{gcd}(x,d)=1$ the following relations hold
\begin{subequations}
\begin{alignat}{2}\label{}
  X^w F_y P_t F_x P_s X^{-w}&=F_y P_t F_x X^{wy/(d-x)}P_{s} X^{-w},\\
  Z^m F_y P_t F_x P_s Z^{-m}&=F_y X^{m/(d-y)} P_t X^{-m/x} F_x P_{s} ,
\end{alignat}
\end{subequations}
and for $x+y=d$ the above two relations become
\begin{subequations}
\begin{alignat}{2}\label{}
  X^w F_y P_t F_x P_s X^{-w}&=F_y P_t F_x P_{s-w},\\
  Z^m F_y P_t F_x P_s Z^{-m}&=F_y  P_{t-m/x}  F_x P_{s}.
\end{alignat}
\end{subequations}
    With property~(\ref{eq:qudpro1}), we find $m=x$,
    and the projective representations are
    $v(g)=X$, $v(h)=Z^x$, and in general $v(g_i,h_j)=X^iZ^{xj}$.
    With Eq.~(\ref{eq:rhox}), we can see that
    the SPT order of a qudit cluster state with $\text{gcd}(x,d)=1$ is labeled by $x$.

Next, we study the case when $x$ and $d$ are not coprime,
i.e., $\text{gcd}(x,d)=s$ for an integer $s$.
Let $d=sb$ and $x=sa$ for integers $a$ and $b$.
Such a cluster state only has bond dimension $b$ since
\begin{equation}\label{eq:b}
|C_{sb}(sa,sb-sa)\rangle=|C_b(a,b-a)\rangle |+_s\rangle^{\otimes N},
\end{equation}
for $N$ sites, which is due to the factorizations
\begin{equation}\label{}
Z_d^r= \mathds{1}_r \otimes Z_b , \; S_d^x=\mathds{1}_s \otimes S_b^a \otimes \mathds{1}_s,
\end{equation}
with subscript denoting the dimension for clarity.
This means there is an $s$-dimensional subspace on each site that undergoes trivial evolution.
    For an arbitrary state $|\psi\rangle_b$, the relation
    \begin{equation}\label{}
      X_d |+_s\rangle |\psi\rangle_b =|+_s\rangle X_b |\psi\rangle_b
    \end{equation}
    means the on-site operator $X_d \otimes \mathds{1}_d$ is equivalent to
    $\mathds{1}_s \otimes X_b \otimes \mathds{1}_d$.
    We can further find the projective representation as
    $v(g)=X_b$ and $v(h)=Z_b^a$.
    As $X_b Z_b^a =e^{i 2\pi a/b} Z_b^a X_b$, and $e^{i 2\pi a/b}=e^{i 2\pi x/d}$, we find
    $v(g) v(h)=e^{i 2\pi x/d} v(h) v(g)$,
    which means the integer $x$ labels the SPT order of $|C_d(x,y)\rangle$
    protected by symmetry $\mathbb{Z}_d\times \mathbb{Z}_d$.
\end{proof}

For MBQC, the virtual space of $|C_{sb}(sa,sb-sa)\rangle$
is from that of $|C_b(a,b-a)\rangle$,
as a result, the computational power
of $|C_{sb}(sa,sb-sa)\rangle$
is equivalent to that of $|C_b(a,b-a)\rangle$.
As $\text{gcd}(a,b)=1$,
the state $|C_b(a,b-a)\rangle$ itself lives in the MNC $\mathbb{Z}_b\times \mathbb{Z}_b$ SPT phase
labeled by the integer $a$,
and as well the SPT order of $|C_b(b-a,a)\rangle$ is labeled by $b-a$.

\subsection{Universality for qudit gates}

In this section we discuss the scheme to implement gates,
and we find that the computational power of a qudit cluster state is indicated by its bond dimension.
Our method is to use local unitary operations to convert $|C_d(x,y)\rangle$ to $|C_d(1,1)\rangle$,
the state generated by uniform application of $S^1$,
for which the MBQC scheme has been established before~\cite{Cla06,ZZX+03}.
\begin{proposition}\label{prop:clustgate}
  Cluster states $|C_d(x,y)\rangle$ have bond dimension $b=d/\text{gcd}(x,d)$,
  and the group of gates that can be simulated in MBQC on such states is $SU(b)$.
\end{proposition}
\begin{proof}
For the case $\text{gcd}(x,d)=1$, the bond dimension is $b=d$.
We first show that cluster states with the same bond dimension are locally unitarily equivalent.
Denote $F_{lk}\equiv F_l F_k$ for $\text{gcd}(l,d)=1$ and $\text{gcd}(k,d)=1$,
and from the following property
\begin{align}\label{eq:Flk}
  F_{lk} Z^{-k} F_{lk}^\dagger = Z^{l}, \;  F_{lk} X^{-l} F_{lk}^\dagger = X^{k},
\end{align}
we find
\begin{subequations}\label{eq:qudpro2}
\begin{alignat}{4}
F_{ax}\; \mathds{1}\; F_{ax}\; \mathds{1} \cdots |C_d(y,x)\rangle &= |C_d(a,b)\rangle,\\
\mathds{1}\; F_{ax}\; \mathds{1} \; F_{ax}\;\cdots |C_d(y,x)\rangle &= |C_d(a,b)\rangle,\\
F_{ax}\; \mathds{1}\; F_{ax}\; \mathds{1} \cdots |C_d(x,y)\rangle &= |C_d(b,a)\rangle,\\
\mathds{1}\; F_{ax}\; \mathds{1} \; F_{ax}\;\cdots |C_d(x,y)\rangle &= |C_d(b,a)\rangle,
\end{alignat}
\end{subequations}
which tells us that qudit cluster states $|C_d(x,y)\rangle$ and $|C_d(a,b)\rangle$
with the same bond dimension $d$
can undergo SPT phase transitions to each other driven by local unitary operators $F_{ax}$,
which obviously do not respect the $\mathbb{Z}_d\times \mathbb{Z}_d$ symmetry.
With the permutation operator $\Pi$~(\ref{eq:permutation}) and its property
\begin{equation}\label{}
  \Pi Z^j \Pi = Z^{-j}, \; \Pi X^j \Pi = X^{-j},
\end{equation}
as a special case of Eq.~(\ref{eq:Flk}),
we find
\begin{subequations}
\begin{alignat}{2}\label{eq:qudpro3}
  \mathds{1}\; \mathds{1}\; \Pi\;  \Pi\; \cdots |C_d(x,d-x)\rangle & = |C_d(x,x)\rangle, \\
  \Pi\;  \mathds{1}\; \mathds{1}\; \Pi\; \cdots |C_d(d-x,x)\rangle &= |C_d(x,x)\rangle.
\end{alignat}
\end{subequations}
As the result, we see that a qudit cluster state $|C_d(x,y)\rangle$
can be locally unitarily converted to $|C_d(1,1)\rangle$.

Then we employ the MBQC scheme on $|C_d(1,1)\rangle$~\cite{Cla06,ZZX+03},
namely, given a hermitian operator basis $\{N_\ell\}$ for $\mathcal{B}(\mathcal{H}_d)$,
any qudit gate takes the form
\begin{equation}\label{eq:d-gate}
  U=\prod_\ell e^{i2\pi\beta_\ell N_\ell/d}, \beta_\ell\in \mathbb{R}.
\end{equation}
It's been proved that~\cite{Cla06},
the Fourier transform $F$ together with diagonal phase gates $D(\vec{c})$,
for $\vec{c}=(1,c_1,\dots,c_{d-1})$ of $d-1$ complex numbers of modulus one,
are a universal set of gates for single qudit computation.
From teleportation, similar with the qubit case,
a gate $D(\vec{c})$ can be simply implemented.

For the case $\text{gcd}(x,d)\neq 1$, the state that can be used in MBQC is
$|C_b(a,b-a)\rangle$ from Eq.~(\ref{eq:b}).
It is clear now it is universal for gates in $SU(b)$.
\end{proof}

In addition, it is worth mentioning the transition between
$|C_d(x,y)\rangle$ and $|C_d(y,x)\rangle$, which is
\begin{subequations}\label{eq:qudpro0}
  \begin{alignat}{2}
    \Pi\; \mathds{1}\; \Pi\; \mathds{1} \cdots |C_d(x,y)\rangle &= |C_d(y,x)\rangle,\\
\mathds{1}\; \Pi\; \mathds{1}\;  \Pi\; \cdots |C_d(x,y)\rangle &= |C_d(y,x)\rangle,
  \end{alignat}
\end{subequations}
as a special case of Eq.~(\ref{eq:qudpro2}).
This shows that states $|C_d(x,y)\rangle$ and $|C_d(y,x)\rangle$
are equivalent up to some permutations of local bases.
This also means for $d=2x$, the on-site operations
$\Pi\otimes \mathds{1}$ and $\mathds{1}\otimes  \Pi$ generate a $\mathbb{Z}_2\times \mathbb{Z}_2$ symmetry,
hence state $|C_d(d/2,d/2)\rangle$ has $\mathbb{D}_d\times \mathbb{D}_d$ symmetry
for $\mathbb{D}_d=\langle x, y| x^d=y^2=(xy)^2=1\rangle$.
In the following we show some examples to illustrate our general results.

\begin{example}
For $d=4$, we can define cluster states
$|C_4(1,3)\rangle$, $|C_4(3,1)\rangle$,
$|C_4(1,1)\rangle$, $|C_4(3,3)\rangle$, and $|C_4(2,2)\rangle$.
Among them, states $|C_4(1,3)\rangle$, $|C_4(3,1)\rangle$,
$|C_4(1,1)\rangle$, and $|C_4(3,3)\rangle$ can translate into each other
by proper uses of the permutation $\Pi$.
Eventually, all of them can end up with the state $|C_4(1,1)\rangle$,
which is universal for two-qubit computation.
The state $|C_4(2,2)\rangle$ can reduce to $|C_2(1,1)\rangle$ and an additional $|+_2\rangle$ for each site, hence is universal for single qubit computation.

For $d=5$, we can define cluster states
$|C_5(1,4)\rangle$, $|C_5(4,1)\rangle$,
$|C_5(2,3)\rangle$, $|C_5(3,2)\rangle$,
$|C_5(1,1)\rangle$, $|C_5(2,2)\rangle$, $|C_5(3,3)\rangle$, and $|C_5(4,4)\rangle$.
Among them, $|C_5(1,4)\rangle$ can translate into $|C_5(2,3)\rangle$
by proper uses of $F_{13}$,
while following similar methods,
other transitions can be driven by proper uses of the permutation $\Pi$ and $F_{13}$.
Eventually, all of them can end up with the state $|C_5(1,1)\rangle$ for the purpose of computation.
In this case, all states have bond dimension five.\hfill
$\blacksquare$
\end{example}

\section{MBQC on AKLT-type states}
\label{sec:gaklt}

\begin{table}[t]
  \centering
  \footnotesize
  \caption{Computational properties of higher-symmetry generalizations of AKLT states.}
  \begin{tabular}{|c||c|c|c|c|}
    \hline
     & $SU(N)$  & $SO(2\ell+1)$  & $SO(2\ell+1)$  & $Sp(2n), n=2^m$ \\ \hline \hline
     On-site irrep & $N^2-1$ &  $2\ell+1$ & $\ell(2\ell+1)$ & $n(2n+1)$ \\ \hline
     Virtual irrep & $({\bm N},\bar{{\bm N}})$ &  $2^\ell$ & $2^\ell$ & $2n$ \\ \hline
     Encoding & qunit &  $\ell$ qubits & $\ell$ qubits & $m+1$ qubits \\ \hline
    Byproduct & $X^iZ^j$  & Pauli & Pauli &  Pauli \\ \hline
    Gate set & $SU(N)$ & $SO(2\ell+1)$ & $SO(2\ell+1)$ & $Sp(2n)$ \\ \hline
    Suc. prob. & $\frac{N}{N+1}\rightarrow 1$ & $\frac{2}{2\ell+1} \rightarrow 0$ & $\frac{2(2\ell-1)}{\ell(2\ell+1)} \rightarrow 0$ &
     $\frac{n+1}{2n+1} \rightarrow \frac{1}{2}$ \\ \hline
  \end{tabular}
  \label{tab:akltstates}
\end{table}

The computation scheme based on the spin-1 $SO(3)$ invariant AKLT state can be generalized.
Here we consider AKLT models with unitary $SU(N)$,
orthogonal $SO(2\ell+1)$, and symplectic $Sp(2n)$ symmetries.
Note there are some subtleties of the actual symmetries,
namely, these symmetries are actually
projective unitary $PSU(N)$, projective spin $PSpin(2\ell+1)$,
and projective symplectic $PSp(2n)$ symmetries.
The reasons will be clear later on when those states are studied,
and for simplicity most of the time our simplified terminology is employed.
The universal resource states we identified are summarized in Table~\ref{tab:akltstates},
displaying the on-site irrep,
virtual space irrep, logical encoding, byproduct, gate set that can be performed,
and success probability for gates as well as projections for initialization or readout.

Compared with qudit cluster states, there are the following two primary differences:
i) The success probability of gates is smaller than one for AKLT-type states,
while it is always one for cluster states.
This also occurs for the spin-1 AKLT state, wherein the probability is $2/3$.
An identity gate (up to byproduct) is executed
if a gate or projection on the virtual space is not successful,
and then one can repeat a measurement on the next physical site
until the desired operation is performed.
On average, this yields an increase of the length of the chain by a factor
of the inverse of the success probability.
ii) The gate set that can be induced by symmetry is the same as the given symmetry group,
different from the case of cluster states.
The dimension of the on-site physical system is not the same as the bond dimension,
while for cluster states they are the same.

In the following, we study the computational universality of various AKLT-type states.
Our strategy is to demonstrate the universality directly,
instead of employing the so-called state reduction scheme~\cite{CDJ+10},
which is to reduce the state to a cluster state (or graph state)
by an efficient procedure~\cite{WAR11,WAR12,WHR14,WR15}.
Our method can reveal the properties of AKLT-type states more clearly.

\subsection{MBQC on \texorpdfstring{$SU(N)$}{Lg} AKLT states}
\label{sec:su-n-aklt}

\subsubsection{Construction of states}

We start from the SPT phase classification,
and then we will construct $SU(N)$ AKLT states that belong to different SPT phases.
The classification of 1D bosonic
SPT phases with continuous symmetry, i.e., on-site Lie group symmetry $G$,
has been established~\cite{DQ13a}, and it has been shown
\begin{equation}\label{}
  H^2(G_\Gamma,U(1))\cong \Gamma=\pi_1(G_\Gamma),
\end{equation}
for $G_\Gamma=G/\Gamma$ and a central subgroup $\Gamma\subset Z(G)$, the center of $G$,
and $\pi_1(G_\Gamma)$ is the fundamental group of $G_\Gamma$.
Different from finite group cases, here a nontrivial $\Gamma$ plays a central role for
the classification of SPT phases.
In this subsection we are interested in the group $SU(N)$ and $\Gamma=\mathbb{Z}_N$,
and the group $SU(N)/\mathbb{Z}_N$ is known as projective $SU(N)$, denoted by $PSU(N)$.
Therefore, $PSU(N)$ invariant MPS are classified to be in $N$ phases according to
\begin{equation}\label{}
  H^2(PSU(N),U(1))=\mathbb{Z}_N,
\end{equation}
with one trivial and $N-1$ nontrivial.
Note that we only consider on-site symmetry (without global symmetry such as time reversal).
As an indicator for SPT phases,
it has been proposed to use the congruence class
\begin{equation}\label{eq:cong}
  [\lambda]:= \sum_{i=1}^r \lambda_i v_i,\; \text{mod} \; M, \; [\lambda] \in \mathbb{Z}_N
\end{equation}
of any irreducible representation $\lambda=(\lambda_1,\dots,\lambda_r)$,
written by $r$-tuples of integrable weights,
for the congruence vector $v=(v_i)$ and integer $M$ determined by the group~\cite{DQ13a}.
For the purpose of MBQC,
states with OBC are preferred since it is required for the input and output of the computation.
For $SU(N)$ AKLT states the left and right edge states can be different,
so in general an arbitrary phase can be labeled by
\begin{equation}\label{}
  ([\lambda_\textsc{l}],[\lambda_\textsc{r}])=(x,N-x),\; 1\leq x\leq N,
\end{equation}
hence just by a single integer $x$.
One may notice an appealing similarity with the case of qudit cluster states,
wherein $x$ represents the order in the entangling gate $S^x$~(\ref{eq:Sx}).
Here in particular, the case $x=y=N/2$ means the irrep of $SU(N)$ on the virtual space is real,
and therefore, the state does not break the spatial inversion symmetry.
The case $x=N$ indicates the trivial phase.
In the following we study states with on-site adjoint irrep,
while the virtual space representation can be different,
e.g., the fundamental irrep or other representations.
The adjoint irrep of $SU(N)$
is a faithful representation of $PSU(N)$,
so the actual symmetry is $PSU(N)$ instead of $SU(N)$.
For instance, it is well identified that
the spin-1 AKLT model has $SO(3)$ symmetry instead of $SU(2)$~\cite{CGL+13,SPC11}.

Given the on-site adjoint irrep, usually denoted as $\bm{N^2-1}$,
we define a state by firstly choosing a value for $x$,
and then picking the irrep $\lambda_\textsc{l}$ ($\lambda_\textsc{r}$)
with the smallest dimension such that $x=[\lambda_\textsc{l}]$ ($N-x=[\lambda_\textsc{r}]$).
For even $N$, $x$ can be chosen to be $1,2,\dots,N/2$,
and for odd $N$, $x$ can be chosen to be $1,2,\dots,(N-1)/2$.
Each resulting AKLT-type state is the ground state of a two-body Hamiltonian
\begin{equation}\label{eq:parent-h}
  H=\sum_i \prod_j \left[ (T_i+T_{i+1})^2-C(\lambda_j) \right]
\end{equation}
with
\begin{equation}\label{}
  \lambda_j \in \{\lambda_\textsc{l}\otimes \lambda_\textsc{r}\}\cap\{(\bm{N^2-1})\otimes(\bm{N^2-1})\}
\end{equation}
constructed via the standard valence bond state (VBS) method~\cite{AKLT88},
and $T_i=(T^a_i)$ ($a=1,\dots,N^2-1$) are $su(N)$ generators in the adjoint representation on site $i$,
$C(\lambda_j)$ is the eigenvalue of $\sum_a (T^a)^2$ for irrep $\lambda_j$~\cite{MUM+14}.

Some $SU(N)$ AKLT states for $N\geq 3$ have been constructed~\cite{GR07,KHK08,MUM+14},
namely, in the VBS picture,
each site starts from a fundamental representation ${\bm N}$ and its conjugate $\bm{\bar{N}}$,
and then they are projected to the adjoint representation $\bm{N^2-1}$
based on the decomposition
\begin{equation}\label{}
\bm{N}\otimes \bm{\bar{N}}=(\bm{N^2-1})\oplus \bm{1},
\end{equation}
and the nearest pair $\bm{N}$ and $\bm{\bar{N}}$ from two neighboring sites
are projected to the singlet part $\bm{1}$.
This leads to two nontrivial SPT states that can be labeled by the pair of their edge modes
$(\bm{N}, \bm{\bar{N}})$ and $(\bm{\bar{N}}, \bm{N})$, respectively,
which relate to each other by a spatial inversion.
The congruence class of $\bm{N}$ is always $1$,
and of $\bm{\bar{N}}$ is always $N-1$,
hence the SPT orders of the two AKLT-type states are labeled by $(1,N-1)$ and $(N-1,1)$.

\subsubsection{Resource states for MBQC}
\label{subsec:vir-fund}

In this section we identify resource $SU(N)$ AKLT states that can be defined above for MBQC.
We obtain the following proposition.
\begin{proposition}\label{prop:sun}
  For $SU(N)$ AKLT state with on-site irrep $\bm{N^2-1}$ and virtual irrep $(\bm{N},\bm{\bar{N}})$,
  the gates on the virtual space induced by on-site projective measurements form $SU(N)$
  with success probability $p=N/(N+1)$ and unitary byproduct $X^iZ^j$.
\end{proposition}
\noindent Note this proposition also holds for the state obtained by its spatial inversion.
The proof will be presented at the end of this section after the following analysis.

\emph{Byproduct.}
We start from identifying the byproduct for state labeled by $(\bm{N}, \bm{\bar{N}})$.
As the on-site dimension is $N^2-1$,
in the MPS form there are $N^2-1$ Kraus operators at each site,
and they are site-independent.
From the property of $su(N)$ algebra,
these Kraus operators can be chosen as the generators,
i.e., generalized Gell-Mann matrices (with formula shown in Eq.~(\ref{eq:Gell})), of $su(N)$ algebra
for state $(\bm{N}, \bm{\bar{N}})$,
and the transpose of the generators for $(\bm{\bar{N}}, \bm{N})$.
For the purpose of computation, however,
Gell-Mann matrices are hermitian yet not unitary,
which would not lead to unitary byproduct.
Despite this, we find that there exists a basis and a basis transformation such that
the set of Gell-Mann matrices, denoted by $\{M_\beta\}$, can be transformed
to the set of Heisenberg-Weyl matrices, denoted by $\{W_\alpha\}$,
\begin{equation}\label{}
  \sum_\beta u_{\alpha\beta} M_\beta = W_\alpha
\end{equation}
for $W_\alpha:=X^iZ^j$ with $\alpha=i+dj$,
and $U=(u_{\alpha\beta})$ is a unitary operator.
This is based on the fact that both $\{M_\beta\}$ and $\{W_\alpha\}$
are traceless orthonormal basis.
The Heisenberg-Weyl operators are unitary and will basically be treated as byproducts in MBQC.

\emph{MPS circuit.}
The quantum circuit to represent an $SU(N)$ AKLT state
can be achieved by a direct generalization of
the circuit to generate the spin-1 AKLT state described in section~\ref{subsec:1aklt}.
Namely, each local subsystem is now $N$-dimensional,
the initial state $|+_3\rangle$ is generalized to be $|+_N\rangle$,
the generalized uniformly-controlled gate now becomes
\begin{equation}\label{}
  U= \sum_{ij} P_{ij}\otimes X^i Z^j,
\end{equation}
and the measurement outcomes $s_1s_2$ correspond to the gate
$X^{s_1}Z^{s_2}$ (except $X^0Z^0=\mathds{1}$)
on the input state $|\text{in}\rangle$.
The qudit version of swap gate can also be directly defined.
In addition, the equivalence between the virtual picture and real picture still holds apparently.

\emph{Gates for computation.}
Furthermore, nontrivial gates are induced by projective measurements in rotated bases,
similar with the cases of qudit cluster states and also spin-1 AKLT state.
To find proper rotated bases for nontrivial gates,
we employ the method develop in Refs.~\cite{ZZX+03,Cla06}.
Let's denote $N\equiv d$ and consider prime $d$ for simplicity
(other cases have also been settled),
the basis $\{N_\ell\}$ in equation~(\ref{eq:d-gate}) is chosen to be a set of $d^2$ hermitian operators
that come from the eigenstates of Heisenberg-Weyl operators
\begin{equation}\label{eq:mub}
  \{Z,X,XZ,X^2Z,\dots,X^{d-1}Z\},
\end{equation}
and the gates that need to be executed on the virtual space take the form
$e^{i\theta |\lambda\rangle\langle\lambda|}$
for $|\lambda\rangle$ as an eigenstate of one operator in the set~(\ref{eq:mub}).
Then an arbitrary qudit gate on the virtual space can be written as
\begin{equation}\label{eq:d-gate2}
  V=\prod_\ell e^{i\theta_\ell |\lambda_\ell\rangle\langle\lambda_\ell|}:=\prod_\ell V_\ell
\end{equation}
for $V_\ell:=e^{i\theta_\ell |\lambda_\ell\rangle\langle\lambda_\ell|}$.
As a result, we only need to implement $d^2$ gates
$e^{i\theta_\ell |\lambda_\ell\rangle\langle\lambda_\ell|}$,
each of which is induced by a projective measurement in a rotated basis by an operator $U_\ell$
acting on the on-site physical system.
Furthermore, each operator $U_\ell$ can be readily found,
since for an arbitrary operator $V\in SU(d)$ in the fundamental representation,
the corresponding operator $U$ in the adjoint representation can be found with its matrix elements
\begin{equation}\label{eq:adjunitary}
  u_{ij}=\text{tr}(W_i^\dagger VW_j V^\dagger)/d.
\end{equation}
This is based on the fact that,
the virtual system state $\rho$ can be expanded in the basis $\{W_j\}$ as
\begin{equation}\label{}
  \rho=(\mathds{1}+\sum_j r_j W_j)/d,
\end{equation}
which after the action of $V$ becomes
\begin{align}\label{}\nonumber
  \rho_f=V\rho V^\dagger & =\frac{1}{d}(\mathds{1}+\sum_j r_j V W_j V^\dagger) \\
  &= \frac{1}{d}(\mathds{1}+\sum_{ij} r_j \text{tr}( W_i^\dagger V  W_j V^\dagger ) W_i).
\end{align}
By defining $s_i:=\sum_{j} r_j \text{tr}( W_i^\dagger V  W_j V^\dagger )$ we have
$\rho_f=(\mathds{1}+\sum_{i} s_i W_i)/d$.
The operator $U$ in the adjoint representation maps the vector $\bm{r}:=(r_j)$ to the vector $\bm{s}:=(s_i)$.

After a sequence of $d^2$ projective measurements in rotated bases each specified by $U_\ell$,
given input state $|\text{in}\rangle$, the final state on the virtual space is
\begin{equation}\label{eq:out-akltn}
  |\text{out}\rangle=\prod_\ell V_\ell^\dagger X^{i_\ell}Z^{j_\ell} V_\ell |\text{in}\rangle,
\end{equation}
as a generalization of Eq.~(\ref{eq:out-aklt}).

\emph{Byproduct propagation and success probability.}
In Eq.~(\ref{eq:out-akltn}) the operators $X^{i_\ell}Z^{j_\ell}$ are byproduct
and need to be pulled out.
If $|\lambda_\ell\rangle$ is an eigenstate of operator $N_\ell$ from the set~(\ref{eq:mub}),
then it is straightforward to find the following decompositions
\begin{equation}\label{}
  V_\ell=\sum_{s=0}^{d-1} v_s N_\ell^s=v_0 \mathds{1}+ v_1 N_\ell + \cdots,
\end{equation}
for proper coefficients $v_s\in \mathbb{C}$.
From this, we observe that the probability for
$V_\ell$ commuting with $X^{i_\ell}Z^{j_\ell}$ is $\frac{d-1}{d^2-1}$,
hence the probability for successfully inducing a nontrivial gate is
\begin{equation}\label{eq:sucaklt}
  \mathfrak{p}_\text{suc.}=\frac{d}{d+1},
\end{equation}
which includes the spin-1 AKLT state as a special case, i.e., for $d=2$.
As pointed out before, for gate computation
this can be dealt with by increasing the length of the chain by a factor of $(d+1)/d$.

\emph{Initialization and Readout (projection on the virtual space).}
For readout and also preparation in the virtual space
we need to choose a basis such that some Kraus operators become projectors.
We observe that a projector can be made from linear combination of Heisenberg-Weyl operators.
This can be easily found in a rotated basis by
\begin{equation}\label{eq:adjproj}
  U=\text{diag}[F,F,\dots,F,\mathds{1}_{d-1}]
\end{equation}
with ($d-1$)-fold Fourier operator $F$ defined in Eq.~(\ref{eq:fk}),
and the resulting operators are the set
\begin{equation}\label{}
\{X^iP_j, Z^k\}
\end{equation}
for $i,k=1,\dots,d-1$, $j=0,\dots,d-1$.
We can see that the success probability of projection is also $\frac{d}{d+1}$ (up to byproduct $X^i$),
the same with the success probability~(\ref{eq:sucaklt}) of inducing nontrivial gate.
Now we can prove the main proposition in this section.
\begin{proof}[Proof of Prop.~\ref{prop:sun}]
  From Eq.~(\ref{eq:d-gate2}) and~(\ref{eq:out-akltn}) the gates that can be performed form the whole $SU(N)$ group,
  and it is clear that the byproduct are $X^iZ^j$.
  The success probability for nontrivial gate or projection is from Eq.~(\ref{eq:sucaklt}).
\end{proof}

Furthermore, we have seen that
there are $N$ SPT phases for $SU(N)$ symmetry and corresponding AKLT states,
however, we only identify two AKLT states for each $N$ as resource for MBQC.
We have evidence that other types of AKLT states may not be resource directly,
as shown by the following examples,
but it is possible that they could also be turned into resource
after a certain manipulation
or by employing a certain generalized notion of universality for MBQC~\cite{RWP+16,SWP+16}.

\begin{example}
For the $SU(3)$ AKLT model with on-site $\bm{8}$ irrep,
two AKLT states are defined by using
$(\bm{3},\bar{\bm{3}})$ and $(\bar{\bm{3}},\bm{3})$ in the virtual space for each site,
and labeled by $(1,2)$ and $(2,1)$, respectively, for their SPT order.
It's been shown that~\cite{MUM+14} for the state ($\bm{3}$, $\bar{\bm{3}}$)
the eight Kraus operators on each site are
$I^\pm, V^\pm, U^\pm, \sqrt{2}T^3, \sqrt{2}T^8$,
for $I^\pm=T^1\pm iT^2$, $U^\pm=T^6\pm iT^7$, $V^\pm=T^4\pm iT^5$,
and $T^a$ are the eight Gell-Mann matrices (up to a constant factor) of $su(3)$ generators.
For computation, the unitary operator that converts the set of Gell-Mann matrices
to the set of Heisenberg-Weyl matrices is
$U=\text{diag}(F,F,E)$
for qutrit Fourier transform $F$ and a unitary operator
\begin{equation}\label{}
  E=\frac{1}{\sqrt{6}}\left(\begin{array}{cc}
      1-\omega^2 & 1-\omega \\
      -\sqrt{3}\omega & \sqrt{3}(1+\omega)
    \end{array}\right)
\end{equation}
and $\omega=e^{i2\pi/3}$.
For projections in the virtual space, given the set of Kraus operators
$\{X$, $XZ$, $XZ^2$, $X^2$, $X^2Z$, $X^2Z^2$, $Z$, $Z^2\}$,
we apply a unitary operator $U=\text{diag}[F,F,\mathds{1}_2]$
for the basis transformation,
obtaining the transformed Kraus operators as
$\{X^2P_0$, $X^2P_2$, $X^2P_1$, $XP_0$, $XP_2$, $XP_1$, $Z$, $Z^2\}$.
The success probability is $3/4$. \hfill $\blacksquare$
\end{example}

\begin{example}
For the $SU(4)$ AKLT model with on-site $\bm{15}$ irrep,
two AKLT states are defined by $(\bm{4},\bar{\bm{4}})$ and $(\bar{\bm{4}},\bm{4})$,
and their labels are $(1,3)$ and $(3,1)$, respectively.
It is clear that these states are universal.

Another AKLT state is defined by the edge mode $\bm{6}$, for which $\bm{6}=\bar{\bm{6}}$.
This state does not break the inversion symmetry,
and its label is $(2,2)$.
From $\bm{6}\otimes \bm{6}=\bm{1}\oplus \bm{15}\oplus \bm{20}'$,
and $\bm{15}\otimes \bm{15}=\bm{1}\oplus 2(\bm{15})\oplus \bm{20}'\oplus \bm{45}\oplus \bar{\bm{45}} \oplus \bm{84}$,
this state is the ground state of its parent Hamiltonian,
defined according to Eq.~(\ref{eq:parent-h}).
For the purpose of MBQC, however, the virtual space is of dimension six,
while the gates that can be induced from symmetry are elements in $SU(4)$.
The byproduct operators are the generators of $su(4)$ in the irrep $\bm{6}$.

For the $SU(5)$ AKLT model with on-site $\bm{24}$ irrep,
in addition to the two states defined by $(\bm{5},\bar{\bm{5}})$ and $(\bar{\bm{5}},\bm{5})$,
with labels $(1,4)$ and $(4,1)$, respectively,
another two states are defined by $(\bm{10},\bar{\bm{10}})$ and $(\bar{\bm{10}},\bm{10})$,
and their labels are $(2,3)$ and $(3,2)$, respectively.
The byproduct operators are the generators of $su(5)$ in the irrep $\bm{10}$.
\hfill $\blacksquare$
\end{example}

\subsection{MBQC on \texorpdfstring{$SO(N)$}{Lg} AKLT states}
\label{sec:so-n}

In this section we discuss another natural generalization of the spin-1 $SO(3)$ AKLT model,
namely, the $SO(N)$ generalizations.
According to SPT phase classification~\cite{HPC+12,DQ13a},
for odd $N=2\ell+1$, there are two phases since
\begin{equation}\label{}
  H^2(SO(2\ell+1),U(1))=\mathbb{Z}_2,
\end{equation}
which means there is only one nontrivial phase for any $\ell\in \mathbb{Z}$.
This is due to $SO(2\ell+1)=Spin(2\ell+1)/\mathbb{Z}_2$
and the fact that the center of $Spin(2\ell+1)$ is $\mathbb{Z}_2$~\cite{DQ13a}.
In the following we study two types of states in the nontrivial phase,
one type is defined with on-site fundamental representation,
and the other type is defined with on-site adjoint representation,
yet it turns out both can have virtual spinor irrep.
Some examples have been constructed before,
such as states with on-site fundamental representation~\cite{TZX08}
and on-site adjoint representation~\cite{SR08,TZX+09}.
Here we are interested in the computational property of those states in the nontrivial phase.

First we review the spinor irrep of $SO(2\ell+1)$ in terms of the
so-called Clifford matrices.
The standard procedure to construct Clifford matrices is recursive and described as follows~\cite{Ram10}.
For $SO(2\ell+1)$, the Clifford matrices are of dimension $2^\ell\times 2^\ell$.
For $\ell=1$, the Clifford matrices are the three Pauli matrices
\begin{equation}\label{}
  \Gamma_1=\sigma_1=X,\; \Gamma_2=\sigma_2=Y,\; \Gamma_3=\sigma_3=Z.
\end{equation}
Given the $2\ell-1$ Clifford matrices $\Gamma_i'$ for $\ell-1$,
the Clifford matrices for the case $\ell$ are
\begin{subequations}
\begin{alignat}{3}\label{}
  \Gamma_i & =\sigma_1\otimes \Gamma_i',\; i=1,2,\dots,2\ell-1, \\
  \Gamma_{2\ell} & =\sigma_2\otimes \mathds{1}, \\
  \Gamma_{2\ell+1} & =\sigma_3\otimes \mathds{1}.
\end{alignat}
\end{subequations}
As a generalization of Pauli operators, they satisfy the anti-commutation relation
\begin{equation}\label{}
  \{\Gamma_a, \Gamma_b\}=2\delta_{ab},\; a,b=1,2,\dots,2\ell+1,
\end{equation}
and the $SO(2\ell+1)$ generators in the spinor form are
\begin{equation}\label{eq:sogenerator}
  \Gamma^{ab}:=[\Gamma^a,\Gamma^b]/2i=-i\Gamma^a \Gamma^b, \; a\neq b.
\end{equation}

For even $N$, the $2^\ell$-dimensional spinor rep of $SO(2\ell+1)$ is reducible:
it can be split into two $2^{\ell-1}$-dimensional spinor irrep of $SO(2\ell)$.
The state with on-site fundamental representation and virtual spinor representation is dimerized~\cite{TZX08},
hence cannot serve as resource for MBQC as the odd $N$ case.

\subsubsection{On-site fundamental representation}
\label{subsec:so_fund}

We first consider states with on-site fundamental irrep,
which would simply be denoted by its dimension.
We obtain the following proposition.
\begin{proposition}\label{prop:son}
  For the $SO(2\ell+1)$ AKLT state with on-site irrep $2\ell+1$ and virtual irrep $2^\ell$,
  the gates on the virtual space induced by on-site projective measurements form $SO(2\ell+1)$
  with success probability $p=2/(2\ell+1)$ and byproducts as (tensor products of) Pauli operators.
\end{proposition}
\indent We follow a similar routine as that for the $SU(N)$ case in the last section.
For systems with on-site fundamental irrep of $SO(2\ell+1)$,
an SPT state is constructed with the virtual space as the spinor irrep,
which is of dimension $2^\ell$~\cite{TZX08}.
We see that the virtual space dimension is exponentially larger than the on-site physical dimension.
In the VBS picture, this model is constructed such that a parent Hamiltonian
\begin{equation}\label{}
  H=\sum_i P_{i,i+1}^s
\end{equation}
is the sum of projectors onto the symmetric part of
\begin{equation}\label{}
  \bm{n}\otimes \bm{n}=\bm{1}\oplus \bm{n(n-1)/2} \oplus \bm{(n+2)(n-1)/2},
\end{equation}
for $\bm{1}$ the singlet part,
$\bm{n(n-1)/2}$ the antisymmetric part,
and $\bm{(n+2)(n-1)/2}$ the symmetric part. The ground state only contains the singlet part and antisymmetric part.
Similar with the spin-1 AKLT state,
this state has a ``diluted'' anti-ferromagnetic string order, long-range correlation, and a gap.
In the MPS form, the Kraus operators are the Clifford matrices
which will be byproduct in MBQC.

We find that the schemes for the MPS quantum circuit,
symmetry induced gates, projection in virtual space, and byproduct operator propagation
are all analogues of those for the qudit cluster states and $SU(N)$ AKLT states.
A general matrix $R\in SO(2\ell+1)$ in the spinor form will be
\begin{equation}\label{}
  V=\prod_{ab}e^{-i \theta_{ab} \Gamma^{ab}}:=\prod_{ab} V_{ab},
\end{equation}
namely, a product of $\ell(2\ell+1)$ unitary matrices $V_{ab}:=e^{-i \theta_{ab} \Gamma^{ab}}$.
Given $V$ acting on the virtual space,
the corresponding gate $U$ acting on the local on-site system can be found based on Eq.~(\ref{eq:adjunitary}),
yet some care is needed.
Here there are only $(2\ell+1)$ byproduct operators $\Gamma^a$,
and the generators $\Gamma^{ab}$ of $so(2\ell+1)$ do not span the virtual space,
which has dimension $2^\ell$, the dimension of $\ell$ qubits.
The spanning basis of $su(2^\ell)$ are tensor product of Pauli matrices,
which include the set $\{\Gamma^{a}\}$ as a subset.
We use the set $\{\Gamma^{a}\}$ and the rest from the basis of $su(2^\ell)$, denoted by $\{P^{a}\}$,
to define the entries $u_{ab}=\text{tr}(G^{a} V G^{b} V^\dagger)/2^\ell$ in Eq.~(\ref{eq:adjunitary}).
Here the set $\{G^{a}\}=\{\Gamma^{a}\}\cup \{P^{a}\}$.
It is easy to find that
$U$ is a $(2\ell+1)\times (2\ell+1)$ orthogonal matrix,
which specifies the measurement bases for nontrivial gates.
Given the input state $|\text{in}\rangle$, the final state on the virtual space is
\begin{equation}\label{eq:out-aklton}
  |\text{out}\rangle=\prod_{ab} V_{ab}^\dagger \Gamma^{j_{ab}} V_{ab} |\text{in}\rangle.
\end{equation}
For nontrivial gates,
we find that the whole set of the group $SO(2\ell+1)$ can be realized on the virtual system.

For byproduct operator propagation, it holds
\begin{subequations}\label{eq:sun-prog}
\begin{alignat}{2}
  e^{-i \theta \Gamma^{ab}} \Gamma^j e^{i \theta \Gamma^{ab}}&=\Gamma^j e^{i 2\theta \Gamma^{ab}}, \; j=a \; \text{or} \; b,\\
  e^{-i \theta \Gamma^{ab}} \Gamma^j e^{i \theta \Gamma^{ab}}&=\Gamma^j,\; j\neq a,b.
\end{alignat}
\end{subequations}
This means the success probability of a nontrivial gate is
\begin{equation}\label{eq:sucakltso}
  \mathfrak{p}_\text{suc.}=\frac{2}{2\ell+1},
\end{equation}
which also includes the spin-1 AKLT state as a special case, i.e., for $\ell=1$.

For projection, it is clear to see that only linear combinations of two Clifford operators can
lead to projection, and they are
$\sigma_1\otimes\cdots\sigma_1\otimes\sigma_1$ and
$\sigma_1\otimes\cdots\sigma_1\otimes\sigma_2$.
With a Hadamard gate $H$, the resulting two projectors are
\begin{align}\label{eq:proj-so}
  P_0 = \sigma_1\otimes\cdots\sigma_1\otimes\sigma^+,
  P_1 = \sigma_1\otimes\cdots\sigma_1\otimes\sigma^-
\end{align}
for $\sigma^{\pm}=(\sigma_1\pm\sigma_2)/2$.
We can see that the success probability of projection is the same with that of nontrivial gates,
and it is smaller than the success probability of $SU(N)$ case.

\begin{proof}[Proof of Prop.~\ref{prop:son}]
Eq.~(\ref{eq:out-aklton}) shows that the gates that can be performed form $SO(2\ell+1)$
  and the byproduct are Clifford matrices, which are tensor product of Pauli operators.
  The success probability for nontrivial gate or projection is from Eq.~(\ref{eq:sucakltso}).
\end{proof}

\begin{example}\label{ex:so5}
For the $SO(5)$ AKLT state with $\ell=2$, physical dimension $d=5$, and virtual dimension $\chi=4$,
there are five Kraus operators at each site
\begin{subequations}\label{eq:so5K}
\begin{alignat}{2}
A_0&=X\otimes X, A_1=X\otimes Y, A_2=X\otimes Z,\\
 A_3&=Y\otimes \mathds{1}, A_4=Z\otimes \mathds{1}.
\end{alignat}
\end{subequations}
The MPS circuit is a direct generalization of the case of spin-1 AKLT state.
In the virtual space picture,
each spin-2 with initial state $|+_5\rangle$ comes from the projection on two four-level systems,
and the uniformly-controlled Pauli gate~(\ref{eq:Uaklt})
is substituted by uniformly controlled ``spinor'' gates
(the five $A_i$ in Eq.~(\ref{eq:so5K})).

For computation the gates that can be realized on the virtual system
form the whole group $SO(5)$.
The virtual system can be treated as two qubits.
A universal set of qubit gates on the second qubit can be induced as follows.
A rotation $e^{-i\theta Z}$ can be achieved with probability $2/5$ by a linear combination of $A_0$ and $A_1$.
Similar results hold for other rotations that together can represent an arbitrary qubit rotation.
On the contrary, a universal set of qubit gates on the first qubit can not be induced.
Only rotations of the form $e^{-i\theta X}$ can be induced with probability $2/5$ by
linear combinations of $A_3$ and $A_4$.
Coupling between the two qubits can be induced by other types of linear combinations of these $A_i$ operators.

There is another way to encode two qubits in the virtual space.
As $SO(4)\cong SU(2)\times SU(2)\subset SO(5)$,
we can use the $SO(4)$ subgroup to induce two set of qubit gates in the virtual space.
This means one level of the on-site physical system is left untouched.
The gates on the virtual space will take the form $V\otimes W\in SU(2)\times SU(2)$,
and the corresponding orthogonal matrix $R\in SO(4)$ can be found by a basis transformation~\cite{FOS07}
\begin{equation}\label{}
R=B(V\otimes W)B^\dagger,
\end{equation}
with
\begin{equation}\label{}
B=\frac{1}{\sqrt{2}}\begin{pmatrix}
1 & 0 & 0 &-i \\ 0 & -i &-1 &0 \\ 0 &-i &1 &0 \\ 1 &0 &0 &i
\end{pmatrix}.
\end{equation}
It turns out the two qubits cannot be coupled together, otherwise it will be able to generate the
whole group $SU(4)$.\hfill $\blacksquare$
\end{example}

\subsubsection{On-site adjoint representation}
\label{subsec:so_adj}

With similar method,
we obtain the following proposition for the case of on-site adjoint irrep.
\begin{proposition}
  For the $SO(2\ell+1)$ AKLT state with on-site irrep $\ell(2\ell+1)$ and virtual irrep $2^\ell$,
  the gates on the virtual space induced by on-site projective measurements form $SO(2\ell+1)$
  with success probability $p=2(2\ell-1)/\ell(2\ell+1)$ and byproducts as (tensor products of) Pauli operators.
\end{proposition}
\noindent The proof is similar with that for Prop.~\ref{prop:son}, hence omitted,
while here the success probability is from Eq.~(\ref{eq:sucaklton2}) below.

With on-site adjoint representation, the set of Kraus operators at each site
can be chosen as the generators $\{\Gamma^{ab}\}$ of the $so(2\ell+1)$ algebra, see Eq.~(\ref{eq:sogenerator}).
Following a similar procedure,
we find the success probability for both nontrivial gates and projections is
\begin{equation}\label{eq:sucaklton2}
  \mathfrak{p}_\text{suc.}=\frac{2(2\ell-1)}{\ell(2\ell+1)}.
\end{equation}
For a nontrivial gate, this is based on the following fact that,
given $\Gamma^a\Gamma^b$, the anti-commutation relation
\begin{equation}\label{eq:proj-so2}
  \{\Gamma^a\Gamma^b , \Gamma^m\Gamma^n\}=0
\end{equation}
holds for the cases $m=a,n\neq a,b$ and $m=b,n\neq a,b$,
which occurs for $2(2\ell-1)$ times out of $\ell(2\ell+1)$ events.
A similar argument can also be made for projection.
This also includes the spin-1 AKLT state as a special case.
Compared with the states with on-site fundamental representation in section~\ref{subsec:so_fund},
the success probability becomes bigger with a factor of $(2\ell-1)/\ell$.

To determine the measurement bases,
we define the matrix entries in  Eq.~(\ref{eq:adjunitary})
with respect to the basis formed by $\{\Gamma^{ab}\}$ and the rest from other operators
in the basis of $su(2^\ell)$.
Now the operator $U$ in the adjoint representation is an $\ell(2\ell+1)\times \ell(2\ell+1)$ orthogonal matrix,
and similar with the case in section~\ref{subsec:so_fund},
this can realize the whole set of $SO(2\ell+1)$ group on the virtual system.

\begin{example}\label{ex:so52}
For the $SO(5)$ state with $\ell=2$, virtual dimension $\chi=4$, and on-site physical dimension $d=10$,
there are ten Kraus operators at each site
\begin{equation}\label{eq:so5k-10}
  \{\mathds{1}X,\mathds{1}Y,\mathds{1}Z,ZX,ZY,ZZ,YX,YY,YZ,X\mathds{1}\}
\end{equation}
defined as $A_iA_j$ with $A_i$ from Eq.~(\ref{eq:so5K}).
Compared with the state in Example~\ref{ex:so5}, the computational properties are similar,
e.g., there are also different ways to use the virtual system.
The advantage of the larger on-site physical dimension is to increase the success probability.
Both rotations $e^{-i\theta X}$ and $e^{-i\theta Z}$ can be induced on the second qubit with probability $3/5$,
while only $e^{-i\theta X}$ can be induced on the first qubit with probability $3/5$.
Projections on both qubits have probability $3/5$. \hfill $\blacksquare$
\end{example}

\subsection{MBQC on \texorpdfstring{$Sp(N)$}{Lg} AKLT states}
\label{sec:symp}

In this last section we study the symplectic symmetry generalizations of spin-1 AKLT state.
In many-body physics symplectic symmetry generalization of $SU(2)$ has been studied, e.g., Refs.~\cite{RS91,WZ05,FC09},
which is a natural way to account for time-reversal symmetry,
and does not require a bipartite lattice, hence it is distinct from the $SU(N)$ models which
rely on a bipartite lattice structure (i.e. fundamental and its adjoint irreps) of the underlying spins.
The symplectic group $Sp(2n)\cong U(2n) \bigcap Sp(2n, \mathbb{C})$ is a compact group
of rank $n$ and dimension $n(2n+1)$, and contains unitary matrices
that preserve the symplectic form
\begin{equation}\label{eq:symform}
\Delta=\begin{pmatrix} 0 & -\mathds{1}\\ \mathds{1} & 0 \end{pmatrix}
\end{equation}
such that
\begin{equation}\label{}
  U^t \Delta U = \Delta, \; \forall U\in Sp(2n).
\end{equation}
Note the group $Sp(2n)$ is also referred as the quaternionic unitary group,
and it is not the same as the non-compact group $Sp(2n, \mathbb{R})$.
Also note the actual symmetry we consider is projective $Sp(2n)$.

The SPT phase classification~\cite{DQ13a} shows that there are two phases
\begin{equation}\label{}
  H^2(Sp(2n)/\mathbb{Z}_2,U(1))=\mathbb{Z}_2
\end{equation}
since the irreps of $Sp(2n)$ only belong to two different congruence classes.
Some examples of VBS with $Sp(2n)$ symmetry have been constructed before~\cite{SR08},
which have on-site adjoint irrep of dimension $n(2n+1)$ and virtual space as the fundamental irrep with dimension $2n$.
Here we study the computational properties of these states for MBQC.
We obtain the following proposition.
\begin{proposition}\label{prop:spn}
  For $Sp(2n)$ AKLT state with on-site irrep $n(2n+1)$ and virtual irrep $2n$ for $n=2^m$,
  the gates on the virtual space induced by on-site projective measurements form $Sp(2n)$
  with success probability $p=(n+1)/(2n+1)$ and byproducts as (tensor products of) Pauli operators.
\end{proposition}
\noindent Again, this claim is proved around the end of this section.

The byproduct operators in MBQC, which will be the generators of $Sp(2n)$~\cite{WZ05}, are determined as follows.
There exists a close relation between the generators of $Sp(2n)$ and $SU(n)$.
The $n(2n+1)$ generators of $Sp(2n)$
can be expressed by the generators of $SU(n)$ and $SU(2)$.
First, let $E_{jk}$ denote the square matrix with a 1 in the $(j,k)$-th entry and 0 elsewhere.
The generalized Gell-Mann matrices for $SU(n)$ are defined as
\begin{subequations}\label{eq:Gell}
\begin{alignat}{3}\label{}
  X_{ij} & =\frac{1}{2}(E_{ij}+E_{ji}),\; 1\leq i<j\leq n, \\
  Y_{ij} & =\frac{-i}{2}(E_{ij}-E_{ji}),\; 1\leq i<j\leq n,  \\
  Z_{j}  & = \frac{1}{\sqrt{2j(j-1)}} \left(\sum_{\ell=1}^{j-1} E_{\ell\ell}-(j-1)E_{jj}\right) ,\; 2\leq j\leq n.
\end{alignat}
\end{subequations}
The generators of $Sp(2n)$ are as follows
\begin{equation}\label{eq:spge}
  \{ X_{ij} \otimes \sigma_k,\; Y_{ij} \otimes \mathds{1}_2, \;
     Z_{j} \otimes \sigma_k, \; \mathds{1}_n \otimes \sigma_k\}:=\{N_\ell\}
\end{equation}
for Pauli matrices $\sigma_k=X,Y,Z$.
The group $Sp(2n)$ has the same dimension with $SO(2n+1)$, and $Sp(4)\cong SO(5)$,
while for $n\geq 3$ they are not locally isomorphic.

A resource wire requires that the byproduct operators are unitary.
We find there exists wire if $n$ is a power of two, i.e., $n=2^m$, $m\in \mathbb{Z}^+$.
In this case, linear combination of operators $X_{ij}$ (also $Y_{ij}$)
can transform them into tensor products of Pauli matrices.
As each $X_{ij}$ ($Y_{ij}$) is symmetric (antisymmetric),
a tensor product of Pauli matrices as well as linear combination of them will also be symmetric (antisymmetric).
Operators $Z_{j}$ are diagonal but not unitary yet can be transformed to tensor products of
$\mathds{1}$ and Pauli $Z$ since $Z_{j}$ are diagonal.
So, to simplify notations, we will just assume the operators $N_\ell$ in Eq.~(\ref{eq:spge})
are already in the form of tensor products of Pauli matrices.

Once we identify the wire, the next problem is to determine
how to encode information in the virtual system and execute gates.
When $n=2^m$, the virtual space is for $m+1$ qubits,
yet not all gates in $SU(2^{m+1})$ can be realized,
and the realizable gates form $Sp(2^{m+1})$,
which contains unitary operators that preserve the symplectic form.
An arbitrary gate on the virtual space takes the form
\begin{equation}\label{eq:out-akltpn}
  V=\prod_\ell e^{-i\theta_\ell N_\ell}
\end{equation}
with $N_\ell$ from Eq.~(\ref{eq:spge}).
Given a gate $V_\ell:=e^{-i\theta_\ell N_\ell}$ on the virtual space,
the corresponding operator $U_\ell$ that determines the on-site measurement basis can be easily found
by noting that,
the basis $\{N_\ell\}$ forms a part of the basis of $su(2^{m+1})$,
and the operator $U_\ell$ can be found to be a $n(2n+1)\times n(2n+1)$ matrix according to Eq.~(\ref{eq:adjunitary}).

To find the success probability of gates, we employ the proof method of induction.
As the operators $X_{ij}$ are symmetric, they only contain an even number of Pauli $Y$ operators.
Also each operator $Y_{ij}$ contains an odd number of Pauli $Y$ operators.
Given the sets $\{X_{ij}\}_m$, $\{Y_{ij}\}_m$, and $\{Z_{j}\}_m$ for $n=2^m$,
the new sets for $n=2^{m+1}$ are
\begin{subequations}\label{eq:spge1}
\begin{alignat}{5}\label{}
  &\{X_{ij}\}_{m+1}\\ \nonumber
  &=\{X_{ij}\otimes \mathds{1}, X_{ij}\otimes X, X_{ij}\otimes Z, Y_{ij}\otimes Y, Z_{j}\otimes X, \mathds{1}_n\otimes X\},\\
  &\{Y_{ij}\}_{m+1}\\ \nonumber
  &=\{Y_{ij}\otimes \mathds{1}, Y_{ij}\otimes X, Y_{ij}\otimes Z, X_{ij}\otimes Y, Z_{j}\otimes Y, \mathds{1}_n\otimes Y\},\\
  &\{Z_{j}\}_{m+1}=\{Z_{j}\otimes \mathds{1}, Z_{j}\otimes Z, \mathds{1}_n\otimes Z\}.
\end{alignat}
\end{subequations}
Given this structure, it is straightforward to prove the following facts by induction:
i) given an arbitrary operator in $\{X_{ij}\}$,
the number of operators in this set that commutes with it is $n^2/4$,
the number of operators in $\{Y_{ij}\}$ that commute with it is $n^2/4-n/2$,
and the number of operators in $\{Z_{j}\}$ that commute with it is $n/2-1$;
ii) given an arbitrary operator in $\{Y_{ij}\}$,
the number of operators in this set that commute with it is $n^2/4$,
the number of operators in $\{X_{ij}\}$ that commute with it is $n^2/4-n/2$,
and the number of operators in $\{Z_{j}\}$ that commute with it is $n/2-1$;
iii) given an arbitrary operator in $\{Z_{j}\}$,
the number of operators in both $\{X_{ij}\}$ and $\{Y_{ij}\}$
that commute with it is $n^2/4-n/2$.
The proof is only based on the commutation relations among these operators.
As it is quite obvious, there is no need to present the details.

Based on the above results,
we find the success probability is the same for each gate of the form $V_\ell$ which is
\begin{equation}\label{eq:susym}
  \mathfrak{p}_\text{suc.}=\frac{n+1}{2n+1} \rightarrow \frac{1}{2}.
\end{equation}
Basically, a gate will succeed if the generator involved anticommutes with a byproduct,
and we obtain identity otherwise.

In addition, as the virtual space is $m+1$ qubits,
one may also consider single qubit gate separately and the coupling between them,
the gate forms and success probability can also be obtained.
Consider the projection on each virtual qubit.
We find that the success probability on each of them is also $\mathfrak{p}_\text{suc.}$.
The proof is also by induction.
Suppose that for $n=2^m$,
the number of projectors on any qubit
from the set $\{X_{ij}\}$ (also $\{Y_{ij}\}$) is $n^2/4-n/2$,
and from the set $\{Z_{j},\mathds{1}_n\}$ is $n$,
which of course holds for $m=1$,
and then based on the set structure~(\ref{eq:spge1}) it is not difficult to find
that for $n=2^{m+1}$,
the number of projectors on any qubit
from the set $\{X_{ij}\}$ (also $\{Y_{ij}\}$) is $n^2-n$,
and from the set $\{Z_{j},\mathds{1}_n\}$ is $2n$,
which leads to success probability $\frac{2n+1}{4n+1}$.
This proves our claim.
Furthermore, with the same method we also find that
qubit gates of the form $e^{i\theta X}$ or $e^{i\theta Z}$ on any qubit
have the same success probability $\mathfrak{p}_\text{suc.}$.
The on-site bases to induce qubit gates
will be simply defined by operators that are a tensor product of $SO(3)$ rotations with identity,
and a tensor product of Hadamard $H$ gates with identity for projections.

\begin{proof}[Proof of Prop.~\ref{prop:spn}]
  From the above analysis,
  the gates that can be performed form $Sp(2^{m+1})$ with gate form
  from Eq.~(\ref{eq:out-akltpn}),
  and the byproducts are tensor products of Pauli operators.
  The success probability for nontrivial gate or projection is from Eq.~(\ref{eq:susym}).
\end{proof}

Furthermore, for the cases when $n$ is not a power of two,
we cannot show that the byproduct can be unitary so far.
For instance, for $n=3$ the $Sp(6)$ AKLT state can be constructed with on-site irrep
$\bm{21}$ and virtual irrep $\bm{6}$.
The virtual space is six-dimensional, while the gates that can be induced by symmetry
form the group $Sp(6)$, which is smaller than $SU(3)$.
This means the virtual space cannot be encoded as a qutrit
for universal qutrit gate computation.
We note, however, more sophisticated study is required for the cases $n\neq 2^m$ in general.
Below are some examples to demonstrate our general results,
and we find there could be alternative encodings for a certain specific case.

\begin{example}
For $n=2$, the state is the same as the $SO(5)$ AKLT state
for on-site $\bm{10}$ irrep since $Sp(4)\cong SO(5)$.
There are ten Kraus operators at each site
\begin{equation}\label{eq:sp4k}
  \{\mathds{1}X,\mathds{1}Y,\mathds{1}Z,ZX,ZY,ZZ,XX,XY,XZ,Y\mathds{1}\},
\end{equation}
which is equivalent to the set~(\ref{eq:so5k-10})
by the unitary transform $S\otimes \mathds{1}$ with phase gate $S$.

Furthermore, we show that there is a rebit encoding of information in the virtual space.
A unitary matrix $U=V+iW\in SU(N)$  with real part $V$ and imaginary part $W$ can be mapped to a matrix
\begin{equation}\label{eq:symplecticS}
S_U=\begin{pmatrix}
V & -W \\ W & V
\end{pmatrix}\in Sp(2N,\mathbb{C}) \bigcap SO(2N),
\end{equation}
which is both symplectic and orthogonal, i.e., $S^t_U \Delta S_U=\Delta$ and $S^t_US_U=\mathds{1}$,
for the symplectic form $\Delta$ defined in Eq.~(\ref{eq:symform}).
The action of $U$ on a state $|\psi\rangle=|\psi_R\rangle+i|\psi_I\rangle$ can also
be expressed by the actions on the real and imaginary parts separately
\begin{equation}\label{}
  |\psi_R\rangle \mapsto V|\psi_R\rangle-W|\psi_I\rangle , \;
  |\psi_I\rangle \mapsto V|\psi_I\rangle+W|\psi_R\rangle.
\end{equation}
This has actually been proposed to be a universal model of quantum computation based on rebits~\cite{RG02}.
Here for our case the virtual system can be treated as two rebits with state of the form
$\begin{psmallmatrix} |\psi_R\rangle \\ |\psi_I\rangle \end{psmallmatrix}$,
which is a two-rebit encoding of a qubit,
and then any qubit gate $U\in SU(2)$ can be performed in terms of its symplectic form $S_U$~(\ref{eq:symplecticS}).
Given $U\in SU(2)$ and its $S_U$,
the corresponding matrix for on-site measurement basis can be simply determined.
Note the set of $S_U\in Sp(4,\mathbb{C}) \bigcap SO(4) \subset Sp(4)$ for qubit gate computation
only forms a subset of the symmetry group $Sp(4)\supset SU(2)$.
Computational basis projection on the virtual space is performed by the same projection
on both the real and imaginary parts,
since a projection $P$ is a real matrix and its rebit version is $\mathds{1}_2 \otimes P$.
\hfill $\blacksquare$
\end{example}

\begin{example}
For $n=4$, the set of generators of $Sp(8)$ are
\begin{align}\label{eq:sp8}\nonumber
\{   \mathds{1} X \sigma_k,  X \mathds{1} \sigma_k, X X \sigma_k, Z X \sigma_k, X Z \sigma_k, Y Y \sigma_k, & \\ \nonumber
    \mathds{1} Y \mathds{1}, Y \mathds{1} \mathds{1}, YX \mathds{1}, ZY \mathds{1}, YZ \mathds{1}, XY \mathds{1}, &  \\
     Z \mathds{1} \sigma_k, \mathds{1} Z \sigma_k, Z Z \sigma_k, \mathds{1} \mathds{1}  \sigma_k\}.
\end{align}
The virtual space is of dimension eight,
hence can be viewed as three qubits.
The unitary gates that can be performed form the group $Sp(8)$,
and the success probability for gate or projection is $5/9$.

For alternative encoding,
we find that the rebit encoding and computation that applies for $n=2$ does not generalize,
namely in this case, $Sp(8)$ does not contain $SU(4)$,
hence the virtual space cannot be used as a rebit encoding of two qubits.
In addition, we also find that the virtual space can be used as three uncoupled qubits
based on the chain relation
\begin{equation}\label{}
Sp(8)\supset SU(2)\times SU(2)\times SU(2),
\end{equation}
which means arbitrary qubit gate can be performed on each of them.
The success probability for operations on each qubit can be obtained
following our general method. \hfill $\blacksquare$
\end{example}

\section{Conclusion}
\label{sec:conc}

In this work, we explored the connection
between measurement-based quantum computation and symmetry
by studying two classes of resource states:
qudit cluster states and AKLT-type states with unitary, orthogonal, or symplectic symmetry,
for which their SPT orders follow from their symmetries.
We find that for both qudit cluster states and AKLT-type states,
the gates that can be performed in the correlation space
follow from the corresponding symmetry group and the SPT index.
Parent Hamiltonians and quantum circuit representations
of these resource states are also analyzed,
while detailed schemes for preparing resource states are a separate topic
and left for other investigations.

We have identified several cases that may not be suitable for MBQC,
which, however, may turn out to be possible if additional manipulations are allowed.
These are AKLT-type states with $SU(N)$ symmetry but the virtual space is not the fundamental irrep,
states with $SO(2\ell)$ symmetry, and states with $Sp(2n)$ symmetry but for $n$ not a power of two.
It is apparent that the reasons we have pointed out for each of them are different,
but whether there is a unified viewpoint, like a group-theoretic proof in Ref.~\cite{RWP+16,SWP+16},
is an interesting open problem.

Our study of cluster states and AKLT-type states can be extended to higher-dimensional cases.
The on-site physical irrep in AKLT-type states will be different from the one-dimensional case,
while the on-site physical dimension for cluster states can stay the same.
Along this direction, examples of 2D $SU(2)$ AKLT states~\cite{WAR11,Miy11}
have been shown to be universal,
while it is left for further investigations
for AKLT states with higher symmetries.

\section{Acknowledgement}

D.-S. Wang would like to thank I. Affleck for helpful discussions.
This work is supported by NSERC and Cifar.
R. R. is scholar of the Cifar Program in Quantum Information Science.

\bibliography{ext}{}

\begin{thebibliography}{10}

\bibitem{RB01}
Robert Raussendorf and Hans~J. Briegel.
\newblock A one-way quantum computer.
\newblock {\em Phys. Rev. Lett.}, 86:5188--5191, May 2001.

\bibitem{NC00}
Michael~A. Nielsen and Isaac~L. Chuang.
\newblock {\em Quantum Computation and Quantum Information}.
\newblock Cambridge University Press, Cambridge U.K., 2000.

\bibitem{GE07}
D.~Gross and J.~Eisert.
\newblock Novel schemes for measurement-based quantum computation.
\newblock {\em Phys. Rev. Lett.}, 98:220503, May 2007.

\bibitem{CZG+09}
Xie Chen, Bei Zeng, Zheng-Cheng Gu, Beni Yoshida, and Isaac~L. Chuang.
\newblock Gapped two-body hamiltonian whose unique ground state is universal
  for one-way quantum computation.
\newblock {\em Phys. Rev. Lett.}, 102:220501, Jun 2009.

\bibitem{CDJ+10}
Xie Chen, Runyao Duan, Zhengfeng Ji, and Bei Zeng.
\newblock Quantum state reduction for universal measurement based computation.
\newblock {\em Phys. Rev. Lett.}, 105:020502, Jul 2010.

\bibitem{WAR11}
Tzu-Chieh Wei, Ian Affleck, and Robert Raussendorf.
\newblock Affleck-kennedy-lieb-tasaki state on a honeycomb lattice is a
  universal quantum computational resource.
\newblock {\em Phys. Rev. Lett.}, 106:070501, Feb 2011.

\bibitem{WAR12}
Tzu-Chieh Wei, Ian Affleck, and Robert Raussendorf.
\newblock Two-dimensional affleck-kennedy-lieb-tasaki state on the honeycomb
  lattice is a universal resource for quantum computation.
\newblock {\em Phys. Rev. A}, 86:032328, Sep 2012.

\bibitem{Miy11}
Akimasa Miyake.
\newblock Quantum computational capability of a 2d valence bond solid phase.
\newblock {\em Ann. Phys.}, 326(7):1656 -- 1671, 2011.
\newblock July 2011 Special Issue.

\bibitem{DBB12}
Andrew~S Darmawan, Gavin~K Brennen, and Stephen~D Bartlett.
\newblock Measurement-based quantum computation in a two-dimensional phase of
  matter.
\newblock {\em New J. Phys.}, 14(1):013023, 2012.

\bibitem{WHR14}
Tzu-Chieh Wei, Poya Haghnegahdar, and Robert Raussendorf.
\newblock Hybrid valence-bond states for universal quantum computation.
\newblock {\em Phys. Rev. A}, 90:042333, Oct 2014.

\bibitem{NW15}
Hendrik~Poulsen Nautrup and Tzu-Chieh Wei.
\newblock Symmetry-protected topologically ordered states for universal quantum
  computation.
\newblock {\em Phys. Rev. A}, 92:052309, Nov 2015.

\bibitem{MM15b}
Jacob Miller and Akimasa Miyake.
\newblock Hierarchy of universal entanglement in 2d measurement-based quantum
  computation.
\newblock {\em npj Quantum Information}, 2:16036, 2016.

\bibitem{HDE+06}
Marc Hein, Wolfgang D{\"u}r, Jens Eisert, Robert Raussendorf, M.~Van~den Nest,
  and H.~J. Briegel.
\newblock Entanglement in graph states and its applications.
\newblock In {\em Quantum Computers, Algorithms and Chaos}, pages 115--218. IOS
  Press, 2006.

\bibitem{CGW11}
Xie Chen, Zheng-Cheng Gu, and Xiao-Gang Wen.
\newblock Classification of gapped symmetric phases in one-dimensional spin
  systems.
\newblock {\em Phys. Rev. B}, 83:035107, Jan 2011.

\bibitem{SPC11}
Norbert Schuch, David P\'erez-Garc\'{i}a, and Ignacio Cirac.
\newblock Classifying quantum phases using matrix product states and projected
  entangled pair states.
\newblock {\em Phys. Rev. B}, 84:165139, Oct 2011.

\bibitem{CGL+13}
Xie Chen, Zheng-Cheng Gu, Zheng-Xin Liu, and Xiao-Gang Wen.
\newblock Symmetry protected topological orders and the group cohomology of
  their symmetry group.
\newblock {\em Phys. Rev. B}, 87:155114, Apr 2013.

\bibitem{DQ13a}
Kasper Duivenvoorden and Thomas Quella.
\newblock Topological phases of spin chains.
\newblock {\em Phys. Rev. B}, 87:125145, Mar 2013.

\bibitem{DQ13b}
Kasper Duivenvoorden and Thomas Quella.
\newblock From symmetry-protected topological order to landau order.
\newblock {\em Phys. Rev. B}, 88:125115, Sep 2013.

\bibitem{NMC+13}
H~Nonne, M~Moliner, Sylvain Capponi, P~Lecheminant, and K~Totsuka.
\newblock Symmetry-protected topological phases of alkaline-earth cold
  fermionic atoms in one dimension.
\newblock {\em EPL (Europhysics Letters)}, 102(3):37008, 2013.

\bibitem{GHG+10}
A.~V. Gorshkov, M.~Hermele, V.~Gurarie, C.~Xu, P.~S. Julienne, J.~Ye,
  P.~Zoller, E.~Demler, M.~D. Lukin, and A.~M. Rey.
\newblock Two-orbital su(n) magnetism with ultracold alkaline-earth atoms.
\newblock {\em Nat. Phys.}, 6(4):289--295, 2010.

\bibitem{NLC+16}
Pierre Nataf, Mikl\'os Lajk\'o, Philippe Corboz, Andreas~M. L\"auchli, Karlo
  Penc, and Fr\'ed\'eric Mila.
\newblock Plaquette order in the su(6) heisenberg model on the honeycomb
  lattice.
\newblock {\em Phys. Rev. B}, 93:201113, May 2016.

\bibitem{VC04}
F.~Verstraete and J.~I. Cirac.
\newblock Valence-bond states for quantum computation.
\newblock {\em Phys. Rev. A}, 70:060302, Dec 2004.

\bibitem{PVW+07}
D.~P\'erez-Garc\'{i}a, F.~Verstraete, M.~M. Wolf, and J.~I. Cirac.
\newblock Matrix product state representations.
\newblock {\em Quant. Inf. Comput.}, 7(5):401--430, 2007.

\bibitem{ZZX+03}
D.~L. Zhou, B.~Zeng, Z.~Xu, and C.~P. Sun.
\newblock Quantum computation based on d-level cluster state.
\newblock {\em Phys. Rev. A}, 68:062303, Dec 2003.

\bibitem{Cla06}
Sean Clark.
\newblock Valence bond solid formalism for d-level one-way quantum computation.
\newblock {\em J. Phys. A: Mathematical and General}, 39(11):2701, 2006.

\bibitem{GR07}
Martin Greiter and Stephan Rachel.
\newblock Valence bond solids for $\mathrm{SU}(n)$ spin chains: Exact models,
  spinon confinement, and the haldane gap.
\newblock {\em Phys. Rev. B}, 75:184441, May 2007.

\bibitem{KHK08}
Hosho Katsura, Takaaki Hirano, and Vladimir~E Korepin.
\newblock Entanglement in an su(n) valence-bond-solid state.
\newblock {\em J. Phys. A: Mathematical and Theoretical}, 41(13):135304, 2008.

\bibitem{MUM+14}
Takahiro Morimoto, Hiroshi Ueda, Tsutomu Momoi, and Akira Furusaki.
\newblock ${Z}_3$ symmetry-protected topological phases in the su(3) aklt
  model.
\newblock {\em Phys. Rev. B}, 90:235111, Dec 2014.

\bibitem{TZX08}
Hong-Hao Tu, Guang-Ming Zhang, and Tao Xiang.
\newblock Class of exactly solvable {SO}(n) symmetric spin chains with matrix
  product ground states.
\newblock {\em Phys. Rev. B}, 78:094404, Sep 2008.

\bibitem{RS91}
N.~Read and S.~Sachdev.
\newblock Large-{N} expansion for frustrated quantum antiferromagnets.
\newblock {\em Phys. Rev. Lett.}, 66(13):1773, 1991.

\bibitem{WZ05}
Congjun Wu and Shou-Cheng Zhang.
\newblock Sufficient condition for absence of the sign problem in the fermionic
  quantum monte carlo algorithm.
\newblock {\em Phys. Rev. B}, 71(15):155115, 2005.

\bibitem{FC09}
R.~Flint and P.~Coleman.
\newblock Symplectic {N} and time reversal in frustrated magnetism.
\newblock {\em Phys. Rev. B}, 79:014424, Jan 2009.

\bibitem{GE10}
D.~Gross and J.~Eisert.
\newblock Quantum computational webs.
\newblock {\em Phys. Rev. A}, 82:040303, Oct 2010.

\bibitem{BM08}
Gavin~K. Brennen and Akimasa Miyake.
\newblock Measurement-based quantum computer in the gapped ground state of a
  two-body hamiltonian.
\newblock {\em Phys. Rev. Lett.}, 101:010502, Jul 2008.

\bibitem{DB09}
Andrew~C. Doherty and Stephen~D. Bartlett.
\newblock Identifying phases of quantum many-body systems that are universal
  for quantum computation.
\newblock {\em Phys. Rev. Lett.}, 103:020506, Jul 2009.

\bibitem{ESB+12}
Dominic~V. Else, Ilai Schwarz, Stephen~D. Bartlett, and Andrew~C. Doherty.
\newblock Symmetry-protected phases for measurement-based quantum computation.
\newblock {\em Phys. Rev. Lett.}, 108:240505, Jun 2012.

\bibitem{PW15}
Abhishodh Prakash and Tzu-Chieh Wei.
\newblock Ground states of one-dimensional symmetry-protected topological
  phases and their utility as resource states for quantum computation.
\newblock {\em Phys. Rev. A}, 92:022310, Aug 2015.

\bibitem{MM15}
Jacob Miller and Akimasa Miyake.
\newblock Resource quality of a symmetry-protected topologically ordered phase
  for quantum computation.
\newblock {\em Phys. Rev. Lett.}, 114:120506, Mar 2015.

\bibitem{Sch11}
Ulrich Schollw{\"o}ck.
\newblock The density-matrix renormalization group in the age of matrix product
  states.
\newblock {\em Ann. Phys.}, 326(1):96--192, 2011.

\bibitem{Oru14}
Rom{\'a}n Or{\'u}s.
\newblock A practical introduction to tensor networks: Matrix product states
  and projected entangled pair states.
\newblock {\em Ann. Phys.}, 349:117--158, 2014.

\bibitem{Sti55}
W.~Forrest Stinespring.
\newblock {Positive Functions on C*-algebras}.
\newblock {\em Proc. Am. Math. Soc.}, 6(2):211--216, Apr 1955.

\bibitem{SSV+05}
C.~Sch\"on, E.~Solano, F.~Verstraete, J.~I. Cirac, and M.~M. Wolf.
\newblock Sequential generation of entangled multiqubit states.
\newblock {\em Phys. Rev. Lett.}, 95:110503, Sep 2005.

\bibitem{FNW92}
Mark Fannes, Bruno Nachtergaele, and Reinhard~F Werner.
\newblock Finitely correlated states on quantum spin chains.
\newblock {\em Commun. Math. Phys.}, 144(3):443--490, 1992.

\bibitem{SWP+09}
Mikel Sanz, Michael~M Wolf, David Perez-Garc{\'\i}a, and J~Ignacio Cirac.
\newblock Matrix product states: Symmetries and two-body hamiltonians.
\newblock {\em Phys. Rev. A}, 79(4):042308, 2009.

\bibitem{GC99}
Daniel Gottesman and Isaac~L Chuang.
\newblock Demonstrating the viability of universal quantum computation using
  teleportation and single-qubit operations.
\newblock {\em Nature}, 402(6760):390--393, 1999.

\bibitem{ZLC00}
Xinlan Zhou, Debbie~W Leung, and Isaac~L Chuang.
\newblock Methodology for quantum logic gate construction.
\newblock {\em Phys. Rev. A}, 62(5):052316, 2000.

\bibitem{AKLT87}
Ian Affleck, Tom Kennedy, Elliott~H. Lieb, and Hal Tasaki.
\newblock Rigorous results on valence-bond ground states in antiferromagnets.
\newblock {\em Phys. Rev. Lett.}, 59:799--802, Aug 1987.

\bibitem{AKLT88}
Ian Affleck, Tom Kennedy, Elliott~H Lieb, and Hal Tasaki.
\newblock Valence bond ground states in isotropic quantum antiferromagnets.
\newblock In {\em Condensed Matter Physics and Exactly Soluble Models}, pages
  253--304. Springer, 1988.

\bibitem{BZ98}
Ya.~G. Berkovich and E.~M. Zhmud.
\newblock {\em Characters of Finite Groups}.
\newblock American Mathematical Society, Providence, Rhode Island, 1998.
\newblock Vol. 1.

\bibitem{EBD13}
Dominic~V. Else, Stephen~D. Bartlett, and Andrew~C. Doherty.
\newblock Hidden symmetry-breaking picture of symmetry-protected topological
  order.
\newblock {\em Phys. Rev. B}, 88:085114, Aug 2013.

\bibitem{WR15}
Tzu-Chieh Wei and Robert Raussendorf.
\newblock Universal measurement-based quantum computation with spin-2
  {A}ffleck-{K}ennedy-{L}ieb-{T}asaki states.
\newblock {\em Phys. Rev. A}, 92:012310, Jul 2015.

\bibitem{RWP+16}
R.~Raussendorf, D.-S. Wang, A.~Prakash, T.-C. Wei, and D.~T. Stephen.
\newblock Symmetry-protected topological phases with uniform computational
  power in one dimension, 2016.
\newblock arXiv:quant-ph/1609.07549.

\bibitem{SWP+16}
D.~T. Stephen, D.-S. Wang, A.~Prakash, T.-C. Wei, and R.~Raussendorf.
\newblock Determining the computational power of symmetry protected topological
  phases, 2016.
\newblock arXiv:quant-ph/1611.08053.

\bibitem{HPC+12}
Jutho Haegeman, David P\'erez-Garc\'{\i}a, Ignacio Cirac, and Norbert Schuch.
\newblock Order parameter for symmetry-protected phases in one dimension.
\newblock {\em Phys. Rev. Lett.}, 109:050402, Jul 2012.

\bibitem{SR08}
Dirk Schuricht and Stephan Rachel.
\newblock Valence bond solid states with symplectic symmetry.
\newblock {\em Phys. Rev. B}, 78:014430, Jul 2008.

\bibitem{TZX+09}
Hong-Hao Tu, Guang-Ming Zhang, Tao Xiang, Zheng-Xin Liu, and Tai-Kai Ng.
\newblock Topologically distinct classes of valence-bond solid states with
  their parent hamiltonians.
\newblock {\em Phys. Rev. B}, 80:014401, Jul 2009.

\bibitem{Ram10}
Pierre Ramond.
\newblock {\em Group theory: a physicist's survey}.
\newblock Cambridge University Press, 2010.

\bibitem{FOS07}
Kazuyuki Fujii, Hiroshi Oike, and Tatsuo Suzuki.
\newblock More on the isomorphism {SU}(2) $\otimes$ {SU}(2) $\cong$ {SO}(4).
\newblock {\em Int. J. Geom. Methods Mod. Phys.}, 4(03):471--485, 2007.

\bibitem{RG02}
Terry Rudolph and Lov Grover.
\newblock A 2 rebit gate universal for quantum computing, 2002.
\newblock arXiv:quant-ph/0210187.

\end{thebibliography}
\bibliographystyle{unsrt}
\end{document}